\documentclass[journal,draftclsnofoot,onecolumn]{IEEEtran}
\usepackage{graphics}
\usepackage{graphicx}
\usepackage{epsfig}
\usepackage{amsmath,amssymb}
\usepackage{mathrsfs,amsfonts,fancyhdr}
\usepackage{multirow}
\usepackage{cite}
\usepackage{amsthm}
\newtheorem{thm}{Theorem}[]
\newtheorem{cor}{Corollary}[]

\theoremstyle{remark}
\newtheorem{rem}[]{Remark}

\begin{document}

\title{Partial Decode-Forward Binning Schemes for the Causal Cognitive Relay Channels}

\author{\authorblockN{Zhuohua Wu and Mai Vu\\}
\thanks{The authors are with the Department of Electrical and Computer Engineering, McGill University, Montreal, Canada
(e-mails: zhuohua.wu@mail.mcgill.ca, mai.h.vu@mcgill.ca).}
}


\maketitle
\begin{abstract}
The causal cognitive relay channel (CRC) has two sender-receiver
pairs, in which the second sender obtains information from the first
sender causally and assists the transmission of both senders. In this
paper, we study both the full- and half-duplex modes. In each mode, we
propose two new coding schemes built successively upon one another to
illustrate the impact of different coding techniques. The first scheme
called \emph{partial decode-forward binning} (PDF-binning) combines
the ideas of partial decode-forward relaying and Gelfand-Pinsker
binning. The second scheme called \emph{Han-Kobayashi partial
decode-forward binning} (HK-PDF-binning) combines PDF-binning with
Han-Kobayashi coding by further splitting rates and applying
superposition coding, conditional binning and relaxed joint decoding.

In both schemes, the second sender decodes a part of the
message from the first sender, then uses Gelfand-Pinsker binning
technique to bin against the decoded codeword, but in such a way that
allows both state nullifying and forwarding. For the Gaussian
channels, this PDF-binning essentializes to a correlation between the
transmit signal and the binning state, which encompasses the
traditional dirty-paper-coding binning as a special case when this
correlation factor is zero. We also provide the closed-form optimal
binning parameter for each scheme.

The 2-phase half-duplex schemes are adapted from the full-duplex ones
by removing block Markov encoding, sending different message parts in
different phases and applying joint decoding across both
phases. Analysis shows that the HK-PDF-binning scheme in both modes
encompasses the Han-Kobayashi rate region and achieves both the
partial decode-forward relaying rate for the first sender and
interference-free rate for the second sender. Furthermore, this scheme
outperforms all existing schemes.
\end{abstract}



\section{Introduction}\label{sec:intro}

\begin{figure*}[t]
\centering
\includegraphics[scale=0.65]{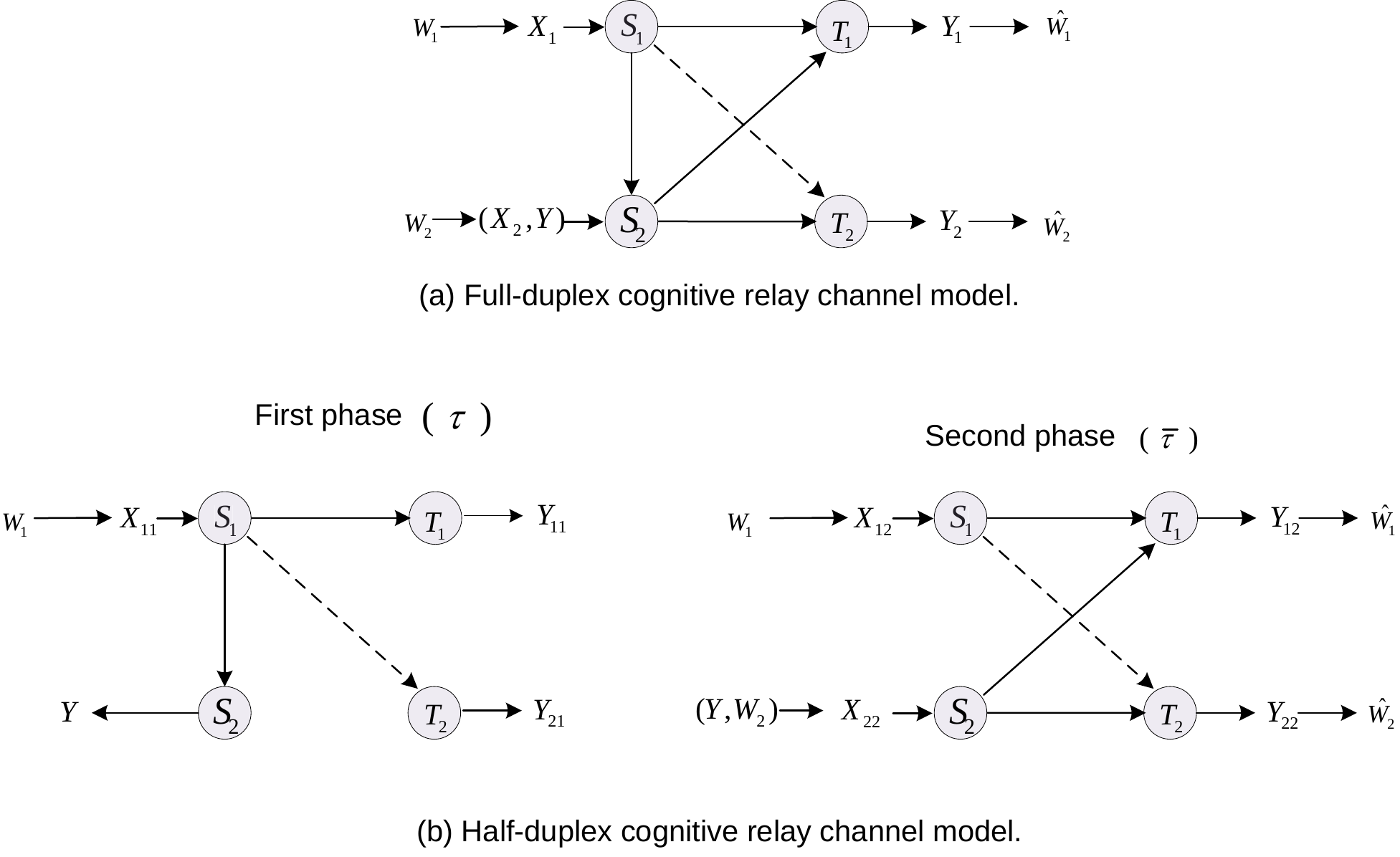}
\caption{The full- and half-duplex modes for the cognitive relay channel.}
\label{fig:FD_and_HD_DM_CRC}
\end{figure*}

\IEEEPARstart{T}{he} causal Cognitive Relay Channel (CRC) is a four-node channel with two senders and two receivers, in which
the second sender obtains information from the first sender causally, then uses that to assist the transmissions of the first sender and its own message.
Different from the assumption in the traditional cognitive channel that the secondary user knows the primary user's message non-causally,
we propose several coding schemes in which the secondary user first decodes the primary user's message causally, then transmits the decoded message and its own message cognitively.

In this paper, we study the cognitive relay channel in both full- and half-duplex modes.
Analysis for the full-duplex mode gives us insights into the optimal coding schemes,
while application to the half-duplex mode is more practical.
In the full-duplex mode, there is no time division into sub-phases;
both senders transmit all messages during the whole transmission.
In the half-duplex mode, however, the transmission is divided into two phases with different message parts sent during each phase.
In the first phase, the second user obtains a message from the first sender causally. In the second phase, these two senders transmit the messages concurrently.

This cognitive relay channel has not been studied much in the literature.
But it has tight relationships with the relay channel (RC), the interference channel (IC) and the cognitive interference channel (CIC).
On the one hand, the second sender serves as a relay and helps forward the message from the first sender.
On the other hand, these two senders interfere with each other during the transmission, and they can also cooperate cognitively.
The closest channel to the CRC is the interference channel with source cooperation (IC with SC), in which both senders can exchange messages causally.

Next, we review existing work related to the cognitive relay channel in both full- and half-duplex modes, then briefly summarize our main results.

\subsection{Related work}
\subsubsection{Full-duplex case}
\emph{i) Relay channel:}
Van der Meulen first proposes the concept of relay channel in \cite{Meulen71}. Cover and El Gamal further design several important techniques for relay channels, including decode-forward, compress-forward, and mixed decode-forward and compress-forward in \cite{Cover79}.
A variation of the decode-forward scheme is partial decode-forward, in which the relay only decodes a part of the message from the source and forwards it to the destination instead of decoding the whole message.
Kramer, Gastpar and Gupta
\cite{Kramer2005TIT} extend these schemes to the multiple-node relay networks and propose several rate regions based on decode-forward, compress-forward and mixed strategies.
Lim, Kim, El Gamal and Chung
\cite{Lim11TIT} propose a new scheme called noisy network coding (NNC) based on compress-forward relaying.
These relay coding techniques have been widely applied in other channels.
For example, in \cite{Liang07a}, Liang and Kramer study the relay broadcast channel using the idea of rate splitting, block Markov encoding and partial decode-forward relaying.

\emph{ii) Interference channel:}
Carleial first introduces the interference channel and proposes inner and outer bounds as well as capacity results for several special cases in \cite{Carleial78}.
Sato studies the capacity for the Gaussian interference channel with strong interference in \cite{Sato81TIT}.
Han and Kobayashi propose the well-known Han-Kobayashi coding technique in \cite{Han} using rate splitting at the transmitters and joint decoding at the receivers, which to date achieves the largest rate region for the interference channel.
Chong, Motani, Garg and El Gamal
\cite{Chong08TIT} propose a variant scheme based on superposition coding, which achieves the same rate region as the original Han-Kobayashi scheme but has fewer auxiliary random variables and hence reduces the encoding and decoding complexities.

\emph{iii) Cognitive interference channel:}
The cognitive interference channel is another closely related channel, which plays a significant role in improving spectrum efficiency.
Devroye, Mitran and Tarokh
first propose the concept in \cite{Devroye2006TIT} and provide an achievable rate region based on combining Gelfand-Pinsker coding \cite{GP1980} with Han-Kobayashi scheme.
They study both the genie-aided (non-causal) and the non genie-aided (causal) cases.
Maric, Yates and Kramer
determine the capacity region for the channel with very strong interference in \cite{Maric2007TIT}.
Wu, Vishwanath and Arapostathis
determine the capacity region for the weak interference case in \cite{Wu2007TIT}.
Other coding schemes for the cognitive interference channel can be seen in \cite{JiangTIT2008, CaoACSSC2008,MaricEuroTele2008}.
Jovicic and Viswanath \cite{Jovicic2009TIT} analyze the Gaussian cognitive channel and give the largest rate for the cognitive user under the constraint
 that the primary user experiences no rate degradation and uses single-user decoder.
Rini, Tuninetti and Devroye \cite{Rini2011TIT}
further propose several new inner, outer bounds and capacity results based on rate spitting, superposition coding, a broadcast channel-like binning scheme and Gelfand-Pinsker coding.

An important technique used in all cognitive coding is the binning technique proposed by Gelfand and Pinsker in \cite{GP1980}.
In Gelfand-Pinsker binning, the state of the channel is known at the input, but unknown at the output.
Marton \cite{Marton79} proposes the double binning scheme and applies it to the broadcast channel.
Kim, Sutivong and Cover \cite{Kim2008TIT} further analyze Gelfand-Pinsker binning to allow the decoding of a part of state information at the destination at a reduced information rate.
Costa \cite{Costa83TIT} applies Gelfand-Pinsker binning to the Gaussian channel and proposes the well-known dirty paper coding (DPC) scheme,
which achieves the same rate as if the channel is interference free.
A surprising feature of DPC binning is that the transmit signal is independent of the state.

\emph{iv) Interference channel with source cooperation:}
Host-Madsen \cite{HostTIT2006} studies outer and inner bounds for the interference channel with either destination or source cooperation.
The achievable rate for source cooperation is based on block Markov encoding and dirty paper coding,
which includes the rate for decode-forward relaying but not the Han-Kobayashi region.
Prabhakaran and Viswanath \cite{PrabhakaranTIT2011} investigate the Gaussian interference channel with source cooperation and propose an achievable rate region built on
block Markov encoding, superposition coding and Han-Kobayashi scheme, but without binning, as well as several upper bounds on the sum rate.
Wang and Tse \cite{WangTIT2011} study the Gaussian interference channel with conferencing transmitters and propose an achievable rate region within 6.5 bits/s/Hz of the capacity for all channel parameters.
The channel is based on conferencing model, in which the common message parts are known though noiseless conference links between the two transmitters before each block transmission begins, hence there is no need for block Markovity.
The scheme utilizes Marton's double binning for the cooperative private messages and superposition coding but not dirty paper coding for the non-cooperative private message parts.
Cao and Chen \cite{CaoISIT2007} propose an achievable rate region for the interference channel with transmitter cooperation
using block Markov encoding, rate splitting and superposition coding, dirty paper coding and random binning.
This scheme includes the Han-Kobayashi region but not the decode-forward relaying rate.
Yang and Tuninetti \cite{YangTIT2011} study the interference channel with generalized feedback (also known as source cooperation)
and propose two schemes.
The first scheme uses rate splitting and block Markov superposition coding only, in which the two users send common messages cooperatively.
The second scheme extends the first one by using both block Markov superposition coding and binning, in which parts of both common and private messages are sent cooperatively.
This scheme also achieves the Han-Kobayashi region but not the decode-forward relaying rate.
We will discuss the schemes in \cite{CaoISIT2007, YangTIT2011} in more details in Section \ref{sec:FD_Tx_coop_compare}.
Tandon and Ulukus \cite{TandonTIT2011} study an outer bound for the MAC with generalized feedback based on dependence balance \cite{HekstraTIT1989} and extend this idea to the interference channel with user cooperation. We will apply this outer bound in Section \ref{sec:FD_Outer_compare}.

\subsubsection{Half-duplex case}
For half-duplex communications, results also exist for the above channels, albeit fewer than in the full-duplex case.

Host-Madsen and Zhang study capacity bounds for the half-duplex relay channel based on time-division in \cite{HostVTC2002, HostTIT2005}
and give achievable rates for the Gaussian relay channel using partial decode-forward and compress-forward.
Zhang, Jiang, Goldsmith and Cui
\cite{Zhang2011TIT} study the half-duplex Gaussian relay channel with arbitrary correlated noises at the relay and destination.
They also evaluate the achievable rates using decode-forward, compress-forward and amplify-forward,
showing none of these schemes strictly outperforms the others.

Peng and Rajan \cite{Peng2009ISIT} study the half-duplex Gaussian interference channel and compute several inner and outer bounds for
transmitter or receiver cooperation. Transmitter cooperation uses decode-forward and divides the transmission into $3$ phases: $2$ broadcast phases and $1$ MIMO cooperative phase.
Wu, Prabhakaran and Viswanath
\cite{Wu2010ISIT} further study source cooperation for the half-duplex interference channel in the symmetric linear deterministic case and compute its sum rate.

For the half-duplex cognitive interference channel,
Devroye, Mitran and Tarokh
\cite{Devroye2006TIT} propose four protocols in which the secondary user obtains the message from the primary user causally.
Time-sharing these 4 protocols can achieve the Han-Kobayashi rate region but not the decode-forward relaying rate.
Chatterjee, Tong and Oyman
\cite{ChatterjeeISIT2010} further propose a new achievable rate region by a 2-phase scheme based on rate splitting, block Markov encoding, Gelfand-Pinsker binning and backward decoding. This scheme can only achieve the rate of decode-forward relaying, which is less than the partial decode-forward rate in the half-duplex mode.
We will discuss these two schemes in more details in Section \ref{sec:numerical_sections}.

\subsection{Summary of Main Results}
In this paper, we fully define the cognitive relay channel in both the full- and half-duplex modes and propose several coding schemes based on
partial decode-forward relaying, Gelfand-Pinsker binning and Han-Kobayashi coding.

\subsubsection{Full-duplex case}
The full-duplex cognitive relay channel is a four-node channel with two sender-receiver pairs $S_1$-$T_1$ and $S_2$-$T_2$, as in Figure \ref{fig:FD_and_HD_DM_CRC}(a).
$S_1$ and $S_2$ want to transmit messages to $T_1$ and $T_2$, respectively.
$S_2$ also serves as a relay by forwarding $S_1$'s message to $T_1$ while transmitting its own message to $T_2$.
Since $S_2$ can both relay and apply cognitive coding at the same time, this gives rise to the name Cognitive Relay Channel (CRC).

We propose two new coding schemes, in which the second scheme is built successively on top of the first one to illustrate the effect of each
technique used.
\begin{itemize}
\item
The first scheme is called \emph{partial decode-forward binning (PDF-binning)}, which utilizes rate splitting, block Markov encoding, partial decode-forward relaying, Gelfand-Pinsker binning and forward joint decoding across two blocks.
$S_1$ divides its message into two parts: one as a private message sent to $T_1$ directly, the other as a forwarding message, which is sent to $T_1$ with the help of $S_2$.
$S_2$ first causally decodes the forwarding message part from $S_1$, then uses the decoded codeword as the binning state.
In this case, however, the binning also allows $S_2$ to forward a part of the state to $T_1$, who then uses joint decoding across two blocks to decode its messages from both $S_1$ and $S_2$.
Different from state amplification in \cite{Kim2008TIT}, here we want to decode the state at a different receiver ($T_1$) from decoding the message ($T_2$).
This scheme achieves the partial decode-forward relaying rate for user 1 and Gelfand-Pinsker rate for user 2.
\item
The second scheme is called \emph{Han-Kobayashi PDF-binning (HK-PDF-binning)}, which combines PDF-binning with Han-Kobayashi coding by having both users further split their messages.
$S_1$ divides its message into three parts: one as the Han-Kobayashi (HK) private message decoded only at $T_1$, another as the HK public message decoded at both $T_1$ and $T_2$, and the final part as the forwarding message.
$S_2$ divides its message into two parts: one as the HK private message and the other as the HK public message.
There are three ideas additional to PDF-binning.
First, in performing partial decode-forward, $S_2$ uses conditional binning instead of traditional binning to bin only its private message part.
Second, although $T_1$ uses joint decoding in both schemes, the decoding rule here is relaxed as $T_1$ also decodes the public message from $S_2$ without requiring it to be correct.
Third, instead of simple Gelfand-Pinsker decoding, $T_2$ uses joint decoding of the binning auxiliary random variable and the HK public messages from the two senders which are encoded independently of the state. HK-PDF-binning achieves both the Han-Kobayashi and the PDF-binning rate regions.

\end{itemize}

\subsubsection{Half-duplex case}
For the half-duplex CRC, the transmission is divided into two phases as in Figure \ref{fig:FD_and_HD_DM_CRC} (b).
In the first phase, $S_1$ transmits to $S_2$, $T_1$ and $T_2$.
In the second phase, the two senders transmit messages simultaneously, during which $S_2$ can both relay and apply cognitive encoding.

We adapt the above two coding schemes to the half-duplex case.
The main challenges in adapting full-duplex schemes to the half-duplex mode include deciding which message parts should be sent in which
phase and changing the destination decoding rule to joint decoding across both phases.

Specifically, we propose two half-duplex (HD) schemes: HD-PDF-binning and HD-HK-PDF-binning.
At the end of the first phase in both schemes, $S_2$ decodes a message part from $S_1$ then applies PDF-binning,
but neither $T_1$ nor $T_2$ decode here. Both $T_1$ and $T_2$ only decode at the end of the second phase.
There are several differences from full-duplex coding.
First, not all message parts are sent in each phase.
Second, there is no need for block Markovity, instead, we superposition codewords in the two phases of the same block.
Third, we use joint decoding at the destinations over two phases of the same block instead of over two consecutive blocks.

\subsubsection{Applications to Gaussian channels}
When applied to the Gaussian channel, a major difference between PDF-binning and the traditional binning in dirty paper coding (DPC) \cite{Costa83TIT}
is that we introduce a correlation between the transmit signal and the state.
This correlation allows both binning and forwarding at the same time, thus helps improve the transmission rate for the first user
and still allows the second user to achieve the interference-free rate.
We derive the closed-form optimal binning parameter for each coding scheme.
This PDF-binning parameter contains the DPC-binning parameter as a special case.

Results show that the HK-PDF-binning scheme outperforms all existing schemes in both the full- and half-duplex modes
for the cognitive relay channel.
Our analysis also shows clearly the impact on rate region of each of the techniques used.
Furthermore, the maximum rate for the primary sender is the rate of partial decode-forward relaying and the maximum rate for the secondary sender is the interference-free rate as in dirty paper coding.

\section{CRC Channel Models}
\subsection{Full-duplex DM-CRC model}
The full-duplex cognitive relay channel consists of two input alphabet $\mathcal{X}_1, \mathcal{X}_2$,
and three output alphabets $\mathcal{Y}_1, \mathcal{Y}_2, \mathcal{Y}$. The channel is
characterized by a channel transition probability $p(y_1,y_2, y|x_1,x_2)$, where $x_1$ and $x_2$ are the
transmit signals of $S_1$ and $S_2$, $y_1$, $y_2$ and $y$ are the received signals of $T_1$, $T_2$ and $S_2$.
Figure \ref{fig:FD_and_HD_DM_CRC}(a) illustrates the channel model, where $W_1$ and $W_2$ are the messages of $S_1$ and $S_2$.
For notation, we use upper case letters to indicates
random variables and lower case letters to indicate their realizations. We use $x^n$ and $x_k^n$ to represent
the vectors ($x_1, \ldots, x_n$) and ($x_k, \ldots, x_n$) respectively.

The cognitive relay channel has tight relationships with
the interference and the relay channels. For example, this channel model can be converted
to the interference channel\cite{Carleial78} if $S_2$ does not
forward any information to $T_1$.
Similarly, this channel reduces to the relay channel \cite{Meulen71}, \cite{Cover79}
if $S_2$ does not have any message for $T_2$.

A ($2^{nR_1}$, $2^{nR_2}$, $n$) code, or a communication strategy for $n$ channel uses with rate pair ($R_{1}, R_2$), consists of the following:
\begin{itemize}
  \item         Two message sets $\mathcal{W}_{1}\times\mathcal{W}_{2}=[1, 2^{nR_1}]\times[1, 2^{nR_2}]$ and independent messages $W_{1}, W_2 $ uniformly distributed over $\mathcal{W}_{1}$ and $\mathcal{W}_{2}$, respectively.
  \item         Two encoders: one maps message $w_{1}$ into codeword $x_1^n(w_1) \in \mathcal{X}_1^n$, and one maps $w_2$ and each received sequence $y^{k-1}$ into a symbol $x_{2k}(w_2,y^{k-1}) \in \mathcal{X}_2$.
  \item         Two decoders:
            one maps $y_1^n$ into  $\hat{w}_{1} \in \mathcal{W}_{1}$;
            one maps $y_2^n$ into $\hat{w}_2 \in \mathcal{W}_2$.
\end{itemize}

The probability of error when the message pair ($W_{1}, W_2$) is sent is defined as $P_e(W_{1}, W_2) = P\{(\widehat{W}_{1}, \widehat{W}_2) \neq (W_{1}, W_2)\}$.
A rate pair ($R_{1}, R_2$) is said to be \emph{achievable} if, for any $\epsilon > 0 $, there exists a code such that the average error probability $P_e\leq \epsilon$ as $n\rightarrow\infty$.
The capacity region is the convex closure of the set of all achievable rate pairs.

\subsection{Half-duplex DM-CRC model}
The half-duplex cognitive relay channel also consists of four nodes: two senders $S_1$, $S_2$ and two receivers $T_1$, $T_2$.
$S_1$ wants to send a message to $T_1$. $S_2$ serves as a causal relay node and helps forward messages from $S_1$ to $T_1$, while also sending its own message to $T_2$.
The transmission in the half-duplex mode is divided into two phases. In the first phase, $S_1$ transmits its message and $S_2$, $T_1$ and $T_2$ listen.
In the second phase, both $S_1$ and $S_2$ transmit and $T_1$ and $T_2$ listen.
This 2-phase transmission allows, for example, $S_2$ to decode a part of the message from $S_1$ in the first phase and then forwards this part with its own message in the second phase.

Formally, the half-duplex cognitive relay channel consists of three input alphabet $\mathcal{X}_{11}$, $\mathcal{X}_{12}$, $\mathcal{X}_{22}$,
and five output alphabets $\mathcal{Y}_{11}$, $\mathcal{Y}_{21}$, $\mathcal{Y}$, $\mathcal{Y}_{12}$, $\mathcal{Y}_{22}$.
The channel is
characterized by a channel transition probability \\$p_c(y_{11},y_{21}, y, y_{12},y_{22},|x_{11},x_{12},x_{22})$ defined as
\begin{align}\label{eq:p_c}
p_c(y_{11},y_{21}, y, y_{12},y_{22},|x_{11},x_{12},x_{22})=
\begin{cases}
                   p(y_{11},y_{21},y|x_{11}) &\quad \text{if   $0 \leq t \leq \tau$}, \\
                   p(y_{12},y_{22}|x_{12},x_{22}) &\quad \text{if   $\tau \leq t \leq 1$},
\end{cases}
\end{align}
where $t$ is the normalized transmission time within $1$ block, $x_{11}$ and $x_{21}$ refer to the
transmit signals of $S_1$ in the first and second phases, respectively;
$x_{22}$ refers to the transmit signal of $S_2$ in the second phase ($S_2$ does not send any signal in the first phase);
$y_{11}$ and $y_{12}$ are the received signals of $T_1$ in the first and second phases;
$y_{21}$ and $y_{22}$ are the received signals of $T_2$ in the two phases;
and $y$ is the received signal of $S_2$ in the first phase.
We assume the channel is memoryless.
Figure \ref{fig:FD_and_HD_DM_CRC}(b) illustrates the channel model, where $W_1$ and $W_2$ are the messages of $S_1$ and $S_2$.
We use the notation $x^{\tau n} = (x_1, x_2, \cdots, x_{\tau n})$ and $x^{\bar{\tau}n} = (x_{\tau n+1}, \cdots, x_n )$,
which correspond to the codewords sent during the first and second phases.

A ($2^{nR_1}$, $2^{nR_2}$, $n$) code, or a communication strategy for $n$ channel uses with rate pair ($R_{1}, R_2$), consists of the following:
\begin{itemize}
  \item         Two message sets $\mathcal{W}_{1}\times\mathcal{W}_{2}=[1, 2^{nR_1}]\times[1, 2^{nR_2}]$ and independent messages $W_{1}, W_2 $ that are uniformly distributed over $\mathcal{W}_{1}$
and $\mathcal{W}_{2}$.
  \item         Three encoders: two that map message $w_{1}$ into  codewords
$x_{11}^n(w_1) \in \mathcal{X}_{11}^n$ and $x_{12}^n(w_1) \in \mathcal{X}_{12}^n$,
and one that maps $w_2$ and $y^{\tau n}$ into a codeword $x_{22}^n(w_2,y^{\tau n}) \in \mathcal{X}_{22}^n$.
  \item         Two decoders:
            One maps $y_1^n$ into  $\hat{w}_{1} \in \mathcal{W}_{1}$;
            and one maps $y_2^n$ into $\hat{w}_2 \in \mathcal{W}_2$.
\end{itemize}

The probability of error, achievable rate and capacity region are defined in a similar way to the full-duplex case.
\section{Full-duplex partial decode-forward binning schemes}
\subsection{PDF-binning scheme}\label{sec:FD_DMC_Bin}
The first scheme uses block Markov superposition encoding at $S_1$ and partial decode-forward relaying and Gelfand-Pinsker binning at $S_2$.
$T_1$ uses joint decoding across two blocks while $T_2$ uses normal Gelfand-Pinsker decoding.
The first sender $S_1$ splits its message $w_1$ into two parts ($w_{10}, w_{11}$), which correspond to the common (forwarding) and private parts.
We use block Markov encoding at $S_1$, such that the current-block common message $w_{10}$ is superimposed on the previous-block common message $w_{10}'$.
Then, message $w_{11}$ is superimposed on both $w_{10}'$ and $w_{10}$.
The second sender $S_2$ decodes the previous common message $w_{10}'$ from the first sender $S_1$ then uses binning to bin against the codeword for this message part.
Depending on the joint distribution between the binning auxiliary random variable and the state that $S_2$ can also forward a part of the state (i.e. message $w_{10}'$) to $T_1$.
The encoding and decoding structure can be seen in Figure \ref{fig:FD_DM_PDF_Binning}, in which $w_{10}'$ corresponds to $w_{10[i-1]}$.

\begin{figure}[t]
\centering
\includegraphics[scale=0.55]{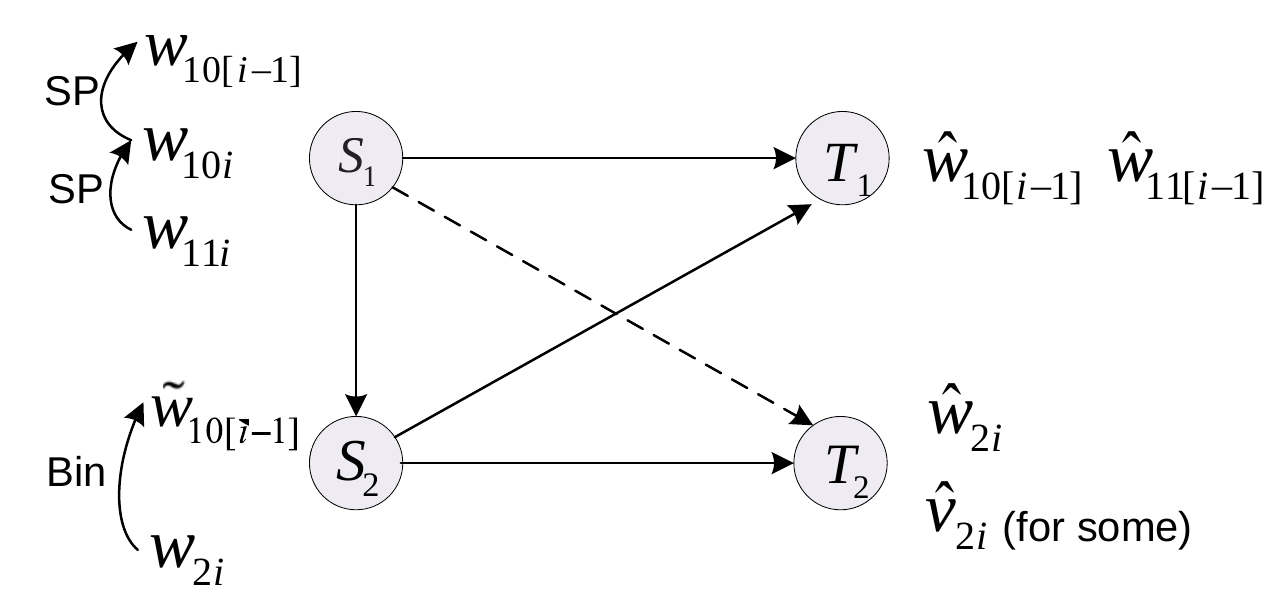}
\caption{Coding structure for the full-duplex PDF-binning scheme at block i. (SP stands for superposition) }
\label{fig:FD_DM_PDF_Binning}
\end{figure}
\begin{thm}
\label{thm:FD_DMC_Bin}
The convex hull of the following rate region is achievable for the full-duplex cognitive relay channel using PDF-binning:
\begin{align}\label{eq:FD_DMC_Bin}
\bigcup_{\substack{P_{1}}}
\left\{\begin{array}{ll}
                   R_{1} &\leq I(U_{10}; Y|T_{10})+ I(X_1;Y_1|U_{10},T_{10})\\
                   R_{1} &\leq I(T_{10}, U_{10} ,  X_1; Y_1)\\
                   R_2   &\leq I(U_2;Y_2) - I( U_2; T_{10})
\end{array}\right.
\end{align}
where
\begin{align*}
    P_{1} = p(t_{10})p(u_{10}|t_{10})p(x_1|t_{10}, u_{10})p(u_2|t_{10}) p(x_2|t_{10},u_2)p(y_1,y_2,y|x_1,x_2).
\end{align*}
\end{thm}
\begin{rem}The maximum rate for each user.
    \begin{itemize}
    \item
    The first user $S_1$ achieves the maximum rate of partial decode-forward relaying if we set $U_2=\emptyset$, $X_2=T_{10}$.
    \begin{align}\label{eq:FD_DMC_Bin_maxR1}
        &R_{1}^{\max} =\max_{\substack{p(u_{10},x_2)p(x_1|u_{10},x_2)}} \min 
        \{I(U_{10};Y|X_2)+I(X_1;Y_1|U_{10},X_2),
                         I(X_1,X_2;Y_1)\}
    \end{align}
In this case, there is no binning but only forwarding at $S_2$.
    \item
    The second user $S_2$ achieves the maximum rate of Gelfand-Pinsker's binning if we set $T_{10}=U_{10}=X_1$.
    \begin{align}\label{eq:FD_DMC_Bin_maxR2}
        R_{2}^{\max} = \max_{\substack{p(x_1,u_2)p(x_2|x_1,u_2)}} \{I(U_2;Y_2) - I(U_2;X_1)\}
    \end{align}
In this case, there is no forwarding of the state at $S_2$.
    \end{itemize}
\end{rem}
\begin{proof}
The transmission is done in $B$ blocks, each consists of $n$ channel uses.
$S_1$ splits each message $w_1$ into two independent parts $(w_{10}, w_{11})$.
During the first $B-1$ blocks, $S_1$ encodes and sends a message tuple ($w_{10[i-1]}, w_{10i},w_{11i}) \in [1,2^{nR_{10}}]\times [1,2^{nR_{10}}]\times [1,2^{nR_{11}}]$;
$S_2$ encodes and sends message ($w_{10[i-1]},w_{2i}) \in [1,2^{nR_{10}}]\times[1,2^{nR_{2}}]$, where $i=1,2,\ldots,B-1$ denotes the block index.
When $B\rightarrow \infty$,
the average rate triple $\left(R_{10}\frac{B-1}{B},R_{11}\frac{B-1}{B},R_2\frac{B-1}{B}\right)$ approaches to ($R_{10},R_{11},R_2$).

We use random codes and fix a joint probability distribution
\begin{align*}
    p(t_{10})p(u_{10}|t_{10})p(x_1|t_{10}, u_{10})p(u_2|t_{10})p(x_2|t_{10},u_2).
\end{align*}

\subsubsection{Codebook generation}
For each block i (We can also just generate two independent codebooks for the odd and even blocks to make the error events of two consecutive blocks independent \cite{Liang07a}.):
\begin{itemize}
  \item
      Independently generate $2^{nR_{10}}$ sequences $t_{10}^n \sim \prod_{k=1}^{n}p(t_{10k})$.
      Index these codewords as  $t_{10}^n(w_{10}')$, $w_{10}'\in [1,2^{nR_{10}}]$.

  \item
      For each $ t_{10}^n( w _{10} ' )$, independently generate $2^{nR_{10}}$ sequences $u_{10}^n \sim \prod_{k=1}^{n}p(u_{10k}|t_{10k})$.
      Index these codewords as $ u_{10}^n( w_{10}|w _{10} ')$, $w_{10} \in [1,2^{nR_{10}}]$.
      $w_{10}$ contains the common message of the current block, while $w_{10}'$ contains the common message of the previous block.

  \item
      For each $ t_{10}^n( w _{10} ' )$ and $ u_{10}^n( w_{10}|w _{10} ')$, independently generate $2^{nR_{11}}$ sequences
      $x_1^n \sim \prod_{k=1}^{n}p(x_{1k}|t_{10k}, u_{10k}$).
       Index these codewords as $x_1^n(w_{11},w_{10}|w_{10}')$, $w_{11} \in [1,2^{nR_{11}}]$, $w_{10} \in [1,2^{nR_{10}}]$.

  \item
      Independently generate $2^{n(R_{2}+ R_{2}')}$ sequences $u_2^n \sim \prod_{k=1}^{n}p(u_{2k})$.
      Index these codewords as  $ u_2^n( w_{2}, v_{2})$, $w_{2} \in [1,2^{nR_{2}}]$
      and  $v_{2} \in [1,2^{nR_{2}'}]$.
  \item For each $t_{10}(w _{10} ')$
      and $ u_2^n( w _{2}, v_{2})$, generate one $x_2^n \sim \prod_{k=1}^{n}p(x_{2k}|t_{10k}, u_{2k})$.
      Denote $x_2^n$ by       $x_2^n(w_{10} ', w_{2}, v_2)$.
\end{itemize}

\subsubsection{Encoding}
At the beginning of block $i$, let ($w_{10i},w_{11i},w_{2i}$) be the new messages to be sent in block $i$, and
($w_{10[i-1]},w_{11[i-1]},w_{2[i-1]}$) be the messages sent in block $i-1$.
\begin{itemize}
  \item
      $S_1$ knows $w_{10[i-1]}$, in order to send ($w_{10i},w_{11i}$), $S_1$ transmits $x_1^n(w_{11i},w_{10i}|w_{10[i-1]})$.
  \item
      $S_2$ searches for a $v_{2i}$ such that
            \begin{align}
                (t_{10}^n{(w_{10[i-1]})}, u_2^n( w_{2i}, v_{2i}) ) \nonumber
                \in A_\epsilon^{(n)}(P_{T_{10} U_2}).
            \end{align}
      Such a $v_{2i}$ exists with high probability if
            \begin{align}\label{eq:FD_DMC_Bin_begin}
                R_{2}' \geq I( U_2; T_{10}).
            \end{align}
      $S_2$ then transmits $x_2^n(w_{10[i-1]}, w_{2i}, v_{2i})$   .
\end{itemize}
\subsubsection{Decoding}
At the end of block $i$:
\begin{itemize}
\item
$S_2$ knows $w_{10[i-1]}$ and declares message $\hat{w}_{10i}$ was sent if it is the unique message
such that
\begin{align*}
    (t_{10}^n{(w_{10[i-1]})}, u_{10}^n( \hat{w}_{10i} | w_{10[i-1]}),
    y^n(i))
    \in A_\epsilon^{(n)}(P_{T_{10} U_{10} Y}),
\end{align*}
where $y^n(i)$ indicates the received signal at $S_2$ in block $i$. We can show that the decoding error probability goes to 0 as $n\rightarrow \infty$ if
\begin{align}
    R_{10}  &\leq I(U_{10}; Y|T_{10}).
\end{align}
\item
$T_1$ knows $w_{10[i-2]}$ and decodes ($w_{10[i-1]},w_{11[i-1]}$) based on the signals received at block $i-1$ and block $i$.
It declares that message pair ($\hat{w}_{10[i-1]}, \hat{w}_{11[i-1]}$) was sent if it is the unique pair such that
\begin{align*}
    ( t_{10}^n{(w_{10[i-2]})}, u_{10}^n(\hat{w}_{10[i-1]} | w_{10[i-2]}),x_1^n(\hat{w}_{11[i-1]},\hat{w}_{10[i-1]}|w_{10[i-2]}),
        y_1^n(i-1))&\in A_\epsilon^{(n)}(P_{T_{10} U_{10} X_1 Y_1})\\
     \text{and}  \quad  (t_{10}^n{(\hat{w}_{10[i-1]})},  y_1^n(i)) &\in A_\epsilon^{(n)}(P_{T_{10} Y_1}).
\end{align*}
The decoding error probability goes to 0 as $n\rightarrow \infty$ if
\begin{align}
             R_{11} &\leq I(X_1;Y_1|U_{10},T_{10}) \nonumber \\
    R_{10} + R_{11} &\leq I(T_{10}, U_{10} ,  X_1; Y_1).
\end{align}
\item
$T_2$ treats $T_{10}$, a part of the signal from $S_1$, as the state and decodes $w_{2i}$ based on the signal received at block $i$.
Specifically, $T_2$ decodes $w_{2i}$ directly using joint typicality between $u_2$ and $y_2$.
It declares that message $\hat{w}_{2i}$ was sent if it is unique such that
\begin{align*}
    (u_2^n(\hat{w}_{2i} ,\hat{v}_{2i}  ), y_2^n(i)) \in   A_\epsilon^{(n)}(P_{U_2 Y_2})
\end{align*}
for some $\hat{v}_{2i}$.
The decoding error probability goes to 0 as $n\rightarrow \infty$ if
\begin{align} \label{eq:FD_DMC_Bin_end}
             R_2 + R_2'  &\leq I(U_2;Y_2).
\end{align}
\end{itemize}
Let $R_1 = R_{10} + R_{11}$, apply Fourier-Motzkin Elimination \cite{Ziegler95} on constraints \eqref{eq:FD_DMC_Bin_begin}-\eqref{eq:FD_DMC_Bin_end}, we get the rate region in \eqref{eq:FD_DMC_Bin}.
\end{proof}

\begin{rem}
While the idea of the basic PDF-binning scheme is straightforward, this scheme allows the understanding of binning to achieve the maximum rates of
partial decode-forward relaying at user 1 as in \eqref{eq:FD_DMC_Bin_maxR1} and Gelfand-Pinsker coding at user 2 as in \eqref{eq:FD_DMC_Bin_maxR2}.
The importance magnifies in the Gaussian application in Section \ref{sec:FD_Gaussian}. This scheme helps build the base for more complicated schemes later.
\end{rem}

\subsection{Han-Kobayashi PDF-binning scheme}\label{sec:FD_DMC_Combined}
\begin{figure}[t]
\centering
\includegraphics[scale=0.55]{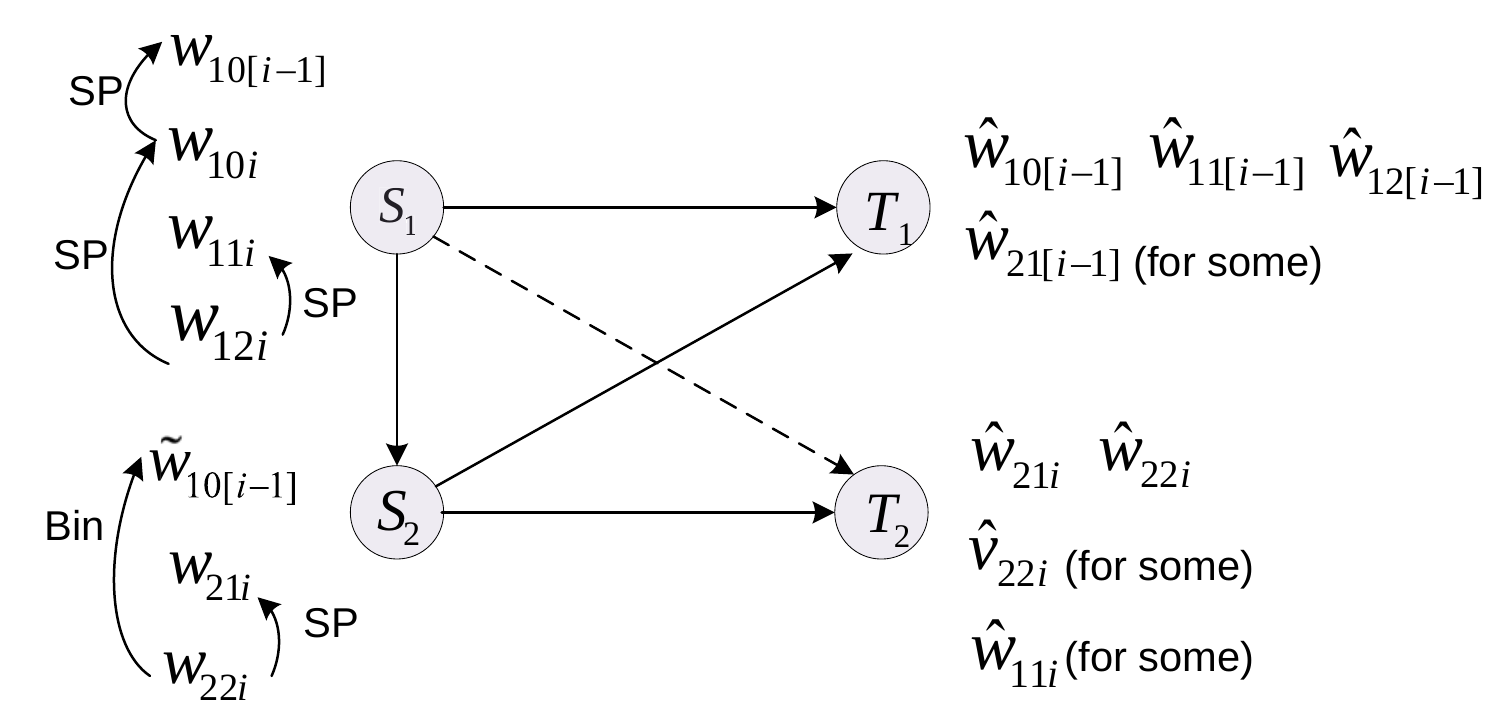}
\caption{Coding structure for the full-duplex Han-Kobayashi PDF-binning scheme at block i. }
\label{fig:FD_DM_Combined}
\end{figure}

Figure \ref{fig:FD_DM_Combined} illustrates the idea of the full-duplex Han-Kobayashi PDF-binning scheme.
Built upon PDF-binning, each user further splits its message to incorporate Han-Kobayashi coding.
Message $w_1$ is split into three parts: $w_{10}, w_{11}, w_{12}$, corresponding to the common (forwarding), public and private parts,
and message $w_2$ is split into two parts: $w_{21}, w_{22}$, corresponding to the public and private parts.
Take the transmission in block $i$ as an example.
At $S_1$, the current common message $w_{10i}$ is superimposed on the previous commons message $w_{10[i-1]}$;
message $w_{11i}$ is encoded independently of both $w_{10[i-1]}$ and $w_{10i}$;
message  $w_{12i}$ is then superimposed on all three messages $w_{10[i-1]}$, $w_{10i}$ and $w_{10i}$.
$S_2$ decodes $\tilde{w}_{10[i-1]}$ of the previous block and uses conditional binning to bin its private part $w_{22[i]}$ against $\tilde{w}_{10[i-1]}$, conditionally on knowing the public part $w_{21[i]}$.
At the end of block $i$, $T_1$ uses joint decoding over two blocks to decode a unique tuple ($\hat{w}_{10[i-1]},\hat{w}_{11[i-1]},\hat{w}_{12[i-1]})$ for some $\hat{w}_{21[i-1]}$
without requiring this message part to be correct.
$T_2$ treats the codeword for $w_{10[i-1]}$ as the state and searches for a unique pair $(w_{21i}, w_{22i})$ for some $w_{11i}$.
The detailed coding and decoding procedures are shown in the proof of Theorem \ref{thm:FD_DMC_combined} below.

\begin{thm}
\label{thm:FD_DMC_combined}
The convex hull of the following rate region is achievable for the cognitive relay channel using HK-PDF-binning:
\begin{align}\label{eq:FD_DMC_combined}
\bigcup_{\substack{P_{2}}}
\left\{\begin{array}{ll}
                             R_{1} &\leq \min\{I_2 + I_5, I_6\}\\
                             R_{2} &\leq I_{12}- I_1 \\
                     R_{1} + R_{2} &\leq \min\{I_2 + I_7,I_8 \} + I_{13}- I_1 \\
                     R_{1} + R_{2} &\leq \min\{I_2 + I_3,I_4\} + I_{14}- I_1 \\
                     R_{1} + R_{2} &\leq \min\{I_2 + I_9, I_{10} \} + I_{11} - I_1 \\
                    2R_{1} + R_{2} &\leq \min\{ I_2 + I_3,I_4\}  
                                   +\min\{I_2 + I_9, I_{10} \} + I_{13}- I_1\\
                    R_{1} + 2R_{2} &\leq \min\{I_2 + I_7, I_8 \}+ I_{11} - I_1 
                                   + I_{14}- I_1
\end{array}\right.
\end{align}
where
\begin{align}
    P_{2} = &p(t_{10})p(u_{10}|t_{10})p(u_{11})p(x_1|t_{10}, u_{10}, u_{11})p(u_{21})
    p(u_{22}|u_{21}, t_{10})p(x_2|t_{10},u_{21},u_{22})p(y_1,y_2,y|x_1,x_2),
\end{align}
and $I_1$ --- $I_{14}$ are defined as
\begin{align}
             I_1 &= I( U_{22}; T_{10}|U_{21}) \nonumber \\
             I_2 &= I(U_{10}; Y|T_{10}) \nonumber \\
             I_3 &= I(X_1;Y_1|T_{10},U_{10}, U_{11}, U_{21}) \nonumber\\
             I_4 &= I(U_{10}, X_1; Y_1|T_{10}, U_{11}, U_{21}) + I(T_{10};Y_1) \nonumber \\
             I_5 &= I(U_{11}, X_1; Y_1|T_{10}, U_{10}, U_{21}) \nonumber \\
             I_6 &= I(U_{10}, U_{11}, X_1; Y_1|T_{10}, U_{21})+ I(T_{10};Y_1) \nonumber \\
             I_7 &= I(X_1, U_{21};Y_1|T_{10},U_{10}, U_{11}) \nonumber \\
             I_8 &=  I(U_{10}, X_1, U_{21}; Y_1|T_{10}, U_{11}) + I(T_{10};Y_1) \nonumber \\
             I_9 &=  I(U_{11}, X_1, U_{21}; Y_1|T_{10}, U_{10}) \nonumber \\
             I_{10} &=  I(T_{10}, U_{10}, U_{11}, X_1, U_{21}; Y_1) \nonumber\\
             I_{11} &=  I(U_{22};Y_2|U_{21},U_{11}) \nonumber \\
             I_{12} &=  I(U_{21},U_{22};Y_2|U_{11})\nonumber \\
             I_{13} &=  I(U_{11},U_{22};Y_2|U_{21})\nonumber \\
             I_{14} &=  I(U_{11},U_{21},U_{22};Y_2).
\end{align}
\end{thm}
\begin{rem}\label{rem:FD_DMC_Combined_Inclusions}
Inclusion of PDF-binning and Han-Kobayashi schemes.
    \begin{itemize}
    \item
    The HK-PDF-binning scheme becomes PDF-binning if $U_{11}=U_{21}=\emptyset$.
    \item
    The HK-PDF-binning scheme becomes the Han-Kobayashi scheme if $T_{10}=U_{10}=\emptyset$ and $X_{2} = U_{22}$.
    \item
    The maximum rates for $S_1$ and $S_2$ are the same as in the PDF-binning scheme in \eqref{eq:FD_DMC_Bin_maxR1} and \eqref{eq:FD_DMC_Bin_maxR2}.
    \end{itemize}
\end{rem}
\begin{proof}
We use random codes and fix a joint probability distribution
\begin{align*}
    p(t_{10})p(u_{10}|t_{10})p(u_{11})p(x_1|t_{10}, u_{10}, u_{11})
    p(u_{21})p(u_{22}|u_{21}, t_{10})p(x_2|t_{10},u_{21},u_{22}).
\end{align*}

\subsubsection{Codebook generation}
For each block i (or for odd and even blocks):
\begin{itemize}
  \item
      Independently generate $2^{nR_{10}}$ sequences $t_{10}^n \sim \prod_{k=1}^{n}p(t_{10k})$.
      Index these codewords as  $t_{10}^n(w_{10}')$, $w_{10}'\in [1,2^{nR_{10}}]$.

  \item
      For each $t_{10}^n(w_{10}')$, independently generate $2^{nR_{10}}$ sequences $u_{10}^n \sim \prod_{k=1}^{n}p(u_{10k}|t_{10k})$.
      Index these codewords as $ u_{10}^n( w_{10}|w_{10} ')$, $w_{10} \in [1,2^{nR_{10}}]$.
      $w_{10}$ is the common message of the current block, while $w_{10}'$ is the common message of the previous block.
  \item
      Independently generate $2^{nR_{11}}$ sequences $u_{11}^n \sim \prod_{k=1}^{n}p(u_{11k})$.
      Index these codewords as  $u_{11}^n(w_{11})$, $w_{11}\in [1,2^{nR_{11}}]$.

  \item
      For each $t_{10}^n(w_{10}')$, $u_{10}^n(w_{10}|w_{10}')$ and $u_{11}^n(w_{11})$, independently generate $2^{nR_{12}}$ sequences\\
      $x_1^n \sim \prod_{k=1}^{n}p(x_{1k}|t_{10k}, u_{10k}, u_{11k}$).
       Index these codewords as $x_1^n(w_{12}|w_{11},w_{10},w_{10}')$, $w_{12} \in [1,2^{nR_{12}}]$.
  \item
      Independently generate $2^{nR_{21}}$ sequences $u_{21}^n \sim \prod_{k=1}^{n}p(u_{21k})$.
      Index these codewords as  $u_{21}^n(w_{21})$, $w_{21}\in [1,2^{nR_{21}}]$.

  \item
      For each $u_{21}^n(w_{21})$, independently generate $2^{n(R_{22}+ R_{22}')}$ sequences $u_{22}^n \sim \prod_{k=1}^{n}p(u_{22k}|u_{21k})$. Index these codewords as $u_{22}^n(w_{22}, v_{22}|w_{21})$, $w_{22} \in [1,2^{nR_{22}}]$
      and  $v_{22} \in [1,2^{nR_{22}'}]$.
  \item For each $t_{10}(w_{10}')$, $u_{21}^n(w_{21})$ and $ u_{22}^n(w_{22}, v_{22}|u_{21})$, generate one $x_2^n \sim \prod_{k=1}^{n}p(x_{2k}|t_{10k}, u_{21i}, u_{22i})$. Denote $x_2^n$ by $x_2^n(w_{10}', w_{21}, w_{22}, v_{22})$.
\end{itemize}

\subsubsection{Encoding}
At the beginning of block $i$, let ($w_{10i},w_{11i},w_{12i},w_{21i},w_{22i}$) be the new messages to be sent in block $i$, and
($w_{10[i-1]},w_{11[i-1]},w_{12[i-1]},w_{21[i-1]},w_{22[i-1]}$) be the messages sent in block $i-1$.
\begin{itemize}
  \item
      $S_1$ knows $w_{10[i-1]}$, in order to send ($w_{10i},w_{11i},w_{12i}$), it transmits $x_1^n(w_{12}|w_{11i},w_{10i},w_{10[i-1]})$.
  \item
      $S_2$ searches for a $v_{22i}$ such that
            \begin{align}\label{eq:FD_DMC_Bin_binstep}
                (t_{10}^n{(w_{10[i-1]})}, u_{21}^n(w_{21i}), u_{22}^n(w_{22i}, v_{22i}|w_{21i}) ) 
                \in A_\epsilon^{(n)}(P_{T_{10} U_{22}|U_{21}}).
            \end{align}
      Such a $v_{22i}$ exists with high probability if
            \begin{align}\label{eq:FD_DMC_combined_begin}
                R_{22}' \geq I( U_{22}; T_{10}|U_{21}).
            \end{align}
      $S_2$ then transmits $x_2^n(w_{10[i-1]},w_{21i}, w_{22i}, v_{22i})$.
\end{itemize}
\subsubsection{Decoding}
At the end of block $i$:
\begin{itemize}
\item $S_2$ knows $w_{10[i-1]}$ and declares message $\hat{w}_{10i}$ was sent if it is the unique message such that
\begin{align*}
    (t_{10}^n{(w_{10[i-1]})}, u_{10}^n(\hat{w}_{10i}|w_{10[i-1]}),y^n(i))
    \in A_\epsilon^{(n)}(P_{T_{10} U_{10} Y}),
\end{align*}
where $y^n(i)$ indicates the received signal at $S_2$ in block $i$. We can show that the decoding error probability goes to 0 when $n\rightarrow \infty$ if
\begin{align}
    R_{10}  &\leq I(U_{10}; Y|T_{10}).
\end{align}
\item $T_1$ knows $w_{10[i-2]}$ and searches for a unique tuple
($\hat{w}_{10[i-1]},\hat{w}_{11[i-1]},\hat{w}_{12[i-1]})$ for some $\hat{w}_{21[i-1]}$ such that
\begin{align}\label{eq:FD_DMC_Combined_Dec_T1}
    ( t_{10}^n{(w_{10[i-2]})}, u_{10}^n(\hat{w}_{10[i-1]}|w_{10[i-2]}),u_{11}^n{(w_{11[i-1]})}, 
    &x_1^n(\hat{w}_{12[i-1]}|\hat{w}_{11[i-1]},\hat{w}_{10[i-1]},w_{10[i-2]}), \nonumber\\
        u_{21}^n(\hat{w}_{21[i-1]}), y_1^n(i-1))&\in A_\epsilon^{(n)}(P_{T_{10} U_{10} U_{11} X_1 U_{21} Y_1})\nonumber\\
     \text{and}  \quad  (t_{10}^n{(\hat{w}_{10[i-1]})},  y_1^n(i)) &\in A_\epsilon^{(n)}(P_{T_{10} Y_1}).
\end{align}
The decoding error probability goes to 0 as $n\rightarrow \infty$ if
\begin{align}
                                         R_{12} &\leq I(X_1;Y_1|T_{10},U_{10}, U_{11}, U_{21}) \nonumber \\
                                R_{10} + R_{12} &\leq I(U_{10}, X_1; Y_1|T_{10}, U_{11}, U_{21}) 
                                                 + I(T_{10};Y_1)\nonumber \\
                                R_{11} + R_{12} &\leq I(U_{11}, X_1; Y_1|T_{10}, U_{10}, U_{21})\nonumber \\
                       R_{10} + R_{11} + R_{12} &\leq I(U_{10}, U_{11}, X_1; Y_1|T_{10}, U_{21})
                                                + I(T_{10};Y_1)\nonumber \\
                                R_{12} + R_{21} &\leq I(X_1, U_{21};Y_1|T_{10},U_{10}, U_{11}) \nonumber \\
                       R_{10} + R_{12} + R_{21} &\leq I(U_{10}, X_1, U_{21}; Y_1|T_{10}, U_{11})  
                                                + I(T_{10};Y_1)\nonumber \\
                       R_{11} + R_{12} + R_{21} &\leq I(U_{11}, X_1, U_{21}; Y_1|T_{10}, U_{10})\nonumber \\
              R_{10} + R_{11} + R_{12} + R_{21} &\leq I(T_{10}, U_{10}, U_{11}, X_1, U_{21}; Y_1).
\end{align}
    \item $T_2$ treats $T_{10}^n{(w_{10[i-1]}')}$ as the state and decodes $(w_{21i}, w_{22i}, v_{22i})$ based on the signal received in block $i$.
    Specifically, $T_2$ searches for a unique $(\hat{w}_{21i}, \hat{w}_{22i})$ for some $(\hat{w}_{11i}, \hat{v}_{22i})$ such that
\begin{align}\label{eq:FD_DMC_Combined_Dec_T2}
    (u_{11}^n(\hat{w}_{11i}), u_{21}^n(\hat{w}_{21i}), &u_{22}^n(\hat{w}_{22i},\hat{v}_{22i}|\hat{w}_{21i}), y_2^n(i)) 
    \in   A_\epsilon^{(n)}(P_{U_{11} U_{21} U_{22} Y_2}).
\end{align}
The decoding error probability goes to 0 as $n\rightarrow \infty$ if
\begin{align}\label{eq:FD_DMC_combined_end}
                     R_{22} + R_{22}' &\leq I(U_{22};Y_2|U_{21},U_{11}) \nonumber \\
            R_{21} + R_{22} + R_{22}' &\leq I(U_{21},U_{22};Y_2|U_{11}) \nonumber \\
            R_{11} + R_{22} + R_{22}' &\leq I(U_{11},U_{22};Y_2|U_{21}) \nonumber \\
   R_{11} + R_{21} + R_{22} + R_{22}' &\leq I(U_{11},U_{21},U_{22};Y_2).
\end{align}
\end{itemize}
Applying Fourier-Motzkin Elimination to \eqref{eq:FD_DMC_combined_begin}-\eqref{eq:FD_DMC_combined_end}, we get rate region \eqref{eq:FD_DMC_combined}.
\end{proof}
\begin{rem}\label{rem:FD_DMC_Combined_features}
Several features of the HK-PDF-binning scheme are worth noting:
\begin{itemize}
    \item
        In encoding, $w_{10}$ and $w_{11}$ are encoded independently, then $w_{12}$ is superpositioned on both.
        This independent coding between the forwarding part ($w_{10}$) and Han-Kobayashi public part ($w_{11}$), rather than superposition, is
        important to ensure the rate region includes both PDF-binning and Han-Kobayashi regions.
    \item
        In the binning step \eqref{eq:FD_DMC_Bin_binstep} at $S_2$, we use conditional binning instead of the usual (unconditional) binning.
        The binning is only between the Han-Kobayashi private message part ($w_{22}$) and the state ($w_{10}'$), conditionally on knowing the Han-Kobayashi
        public messsage part $w_{21}$. This conditional binning is possible since $w_{21}$ is decoded at both destinations.
    \item
        In the decoding step \eqref{eq:FD_DMC_Combined_Dec_T2} at $T_2$, we use joint decoding of both the Gelfand-Pinsker auxiliary random variable ($u_{22}$) and the Han-Kobayashi public message parts ($w_{11}$ and $w_{21}$), instead of decoding Gelfand-Pinsker and Han-Kobayashi codewords separately.
        This joint decoding  is possible since the codewords for $w_{11}$ and $w_{21}$ (i.e. $u_{11}^n$ and $u_{21}^n$) are independent of the state in
        Gelfand-Pinsker coding (i.e. $t_{10}^n$). Joint decoding at both $T_1$ \eqref{eq:FD_DMC_Combined_Dec_T1} and $T_2$ \eqref{eq:FD_DMC_Combined_Dec_T2} help achieve the largest rate region for this coding structure.
\end{itemize}
\end{rem}
\subsection{Comparison with existing schemes for the interference channel with source cooperation}\label{sec:FD_Tx_coop_compare}
In this section, we analyze in detail two existing schemes \cite{CaoISIT2007, YangTIT2011} for the interference channel with source cooperation,
which are most closely related to the proposed schemes.
The interference channel with source cooperation is a 4-node channel in which both $S_1$ and $S_2$ can receive signal from each other and use that cooperatively in
sending messages to $T_1$ and $T_2$. This channel therefore includes the CRC as a special case (when $S_2$ sends no information to $S_1$).
\subsubsection{IC with conferencing}
Cao and Chen \cite{CaoISIT2007} propose an achievable rate region for the interference channel with source cooperation based on rate splitting, block Markov encoding, superposition encoding, dirty paper coding and random binning. Each user splits its message into three parts: common, private and cooperative messages and divides the cooperative message into cells.
The second user generates independent codewords for the current common message
($u_2^n$), previous cooperative cell index ($s_2^n$) and current cooperative message ($w_2^n$).
The codewords for the current private message are then superimposed on the current common message and previous cooperative cell index ($v_2^n|u_2^n, s_2^n$).
Then, the first user treats the previous cooperative-cell-index codeword ($s_2^n$) as the state and jointly bins its codewords for the current common message ($n_1^n$),
previous cooperative cell index ($h_1^n$) and current cooperative message ($g_1^n$).
Finally, the codewords for the first user's private message ($m_1^n$) is conditionally binned with $s_2^n$ given $n_1^n$ and $h_1^n$.
A two-step decoding with list decoding is then used at each destination.

The common, private and cooperative message parts in \cite{CaoISIT2007} correspond roughly to our HK public, HK private and forwarding (common) part, respectively.
As such, when applied to the CRC, their scheme differs from the proposed HK-PDF-binning scheme in the following aspects:
\begin{itemize}
    \item
        Block Markovity is applied only on the HK private part, whereas in our scheme, block Markovity is applied on
        all message parts.
    \item
        Block Markovity is based on cell division of the previous cooperative message, while in our scheme, block Markovity is on the whole previous common message. This, however, is a minor difference since if each cell contains only one message, then cell index reduces to message index.
    \item
        The first user bins both its HK public and private parts (the user labels are switched in \cite{CaoISIT2007}), whereas we only bin the HK private part (see Remark \ref{rem:FD_DMC_Combined_features}).
    \item
        The scheme in \cite{CaoISIT2007} cannot achieve the decode-forward relaying rate because of no block Markovity between the current cooperative-message codeword ($w_2^n$) and the previous cooperative-cell codeword ($s_2^n$).
        In other words, there is no coherent transmission between the source and relay, which can be readily verified from the code distribution.
        Consider setting $V_1 = V_2 = U_1 = U_2 = 0$, $M_1 = M_2 = N_1 = N_2 = 0$ and $W_2 = S_2 = G_2 = H_2 = 0$ in equation (8) of \cite{CaoISIT2007}, then the code distribution reduces to
        \begin{align*}
            p(q)p(g_{1}|q)p(h_{1}|q)p(x_{1}|g_{1},q)p(x_{2}|h_{1},q) = p(q,g_{1},x_{1})p(q,h_{1},x_{2}) \neq p(q, x_{1}, x_{2} ),
        \end{align*}
        where $q$ is the time sharing variable.
        This distribution implies that the first user splits its message into two parts and independently encodes each of them (by $g_1$ and $h_1$).
        The second user then decodes one part in $g_1$ and forwards this part to the destination.
        But because of the independence between $g_1$ and $h_1$, the achievable rate is less than in coherent decode-forward relaying.

        Thus, the claim in Remark 2 of \cite{CaoISIT2007} that this scheme achieves the capacity region of the degraded relay channel is in fact unfounded.
\end{itemize}

\subsubsection{IC with generalized feedback}
Yang and Tuninetti \cite{YangTIT2011} propose two schemes for the interference channel with generalized feedback based on block Markov superposition coding, binning and backward decoding.
Since the first scheme is a special case of the second, we only analyze their second scheme.
Each user splits its message into four parts: cooperative common ($w_{10c}$), cooperative private ($w_{11c}$),
non-cooperative common ($w_{10n}$) and non-cooperative private ($w_{11n}$).
Consider the transmission in block $b$.
First, generate independent codewords for the previous cooperative-common messages of both users
($Q^n(w_{10c,b-1}, w_{20c,b-1})$).
Then the cooperative-common ($w_{10c,b}$), non-cooperative common ($w_{10n,b}$) and non-cooperative private ($w_{11n,b}$) messages
are superimposed on each other successively as $V_1$, $T_1$, $U_1$, respectively (according to $p(v_1, t_1, u_1|q)$).
There are three binning steps after the above codebook generation.
First, the codewords $S_1$, $S_2$ for the previous cooperative-private messages of both users are binned with each other given $Q$.
Second, $V_1$, $U_1$ and $T_1$ are binned with $S_1$ and $S_2$ given $Q$.
Third, the codeword $Z_1$ for the cooperative-private message ($w_{11c,b}$) is conditionally binned with $S_2$, $U_1$ and $T_1$ given $V_1$, $S_1$ and $Q$.
Backward decoding is used, in which each destination applies relaxed joint decoding of all interested messages.


The non-cooperative messages in \cite{YangTIT2011} correspond to our HK public and private parts.
Their scheme has two cooperative message parts (the common is decoded at both destinations while the private is not), whereas the proposed HK-PDF-binning has only one common part.
To compare these two schemes, we consider the following two special settings to make the message parts equivalent:

i) Set the cooperative-common message ($w_{10c}$) to $\emptyset$: Their cooperative private message then corresponds to our forwarding (common) message. Their scheme differs markedly from HK-PDF-binning as follows.
\begin{itemize}
    \item
        User 1 uses binning among the three message parts instead of superposition coding as in HK-PDF-binning.
        Block Markov superposition is also replaced by binning with the codeword for the previous cooperative message.
    \item
        User 2 applies joint binning of both the non-cooperative common and private parts instead of conditional binning of only the non-cooperative private part, given the non-cooperative common part (see Remark \ref{rem:FD_DMC_Combined_features}).
\end{itemize}

ii) Set the cooperative-private message ($w_{11c}$) to $\emptyset$: Their cooperative common message then corresponds to our forwarding (common) message. Their scheme is more similar to HK-PDF-binning, but there are several important differences as follows.
\begin{itemize}
    \item
        User 1 now uses superposition coding, but superimposes all three message parts successively, whereas we generate codewords for the forwarding part and the HK public part independently (see Remark \ref{rem:FD_DMC_Combined_features}).
    \item
        User 2 also applies joint binning of both non-cooperative message parts instead of conditional binning, similar to case i).
    \item
        Destination 2 decodes the cooperative-common part of user 1, thus limits the rate of user 1 to below the decode-forward relaying rate because of the extra rate constraint at destination 2 (this applies even with relaxed decoding).
        In our proposed scheme, the forwarding part of user 1 is not decoded at destination 2.
\end{itemize}
As a result, both schemes in \cite{CaoISIT2007} and \cite{YangTIT2011}, when applied to the CRC, achieve the Han-Kobayashi region but not the decode-forward relaying rate for the first user. Thus, the maximum rates for user 1 in both schemes are smaller than in \eqref{eq:FD_DMC_Bin_maxR1}.

Another point is that, in both \cite{CaoISIT2007} and \cite{YangTIT2011}, joint decoding of both the state and the binning auxiliary random variables
is used at the destinations, but this joint decoding is invalid and results in a rate region larger than is possible.
In our proposed scheme, all message parts that are jointly decoded with the binning auxiliary variable at the second destination are encoded independently of the state.
\begin{rem}
    Based on our analysis, we conjecture that splitting the common (forwarding) message further into two parts is not necessary for the CRC.
    In \cite{YangTIT2011, WangTIT2011}, the common message is split into two parts: one for decoding at the other destination and the other for binning.
    Our analysis shows that both these operations can be included in one-step binning by varying the joint distribution between the state and the auxiliary
    random variable. This joint distribution becomes apparent when applying to the Gaussian channel as in Section \ref{sec:FD_Gaussian} next.
\end{rem}

\begin{figure}[t]
\centering
\includegraphics[scale=0.45]{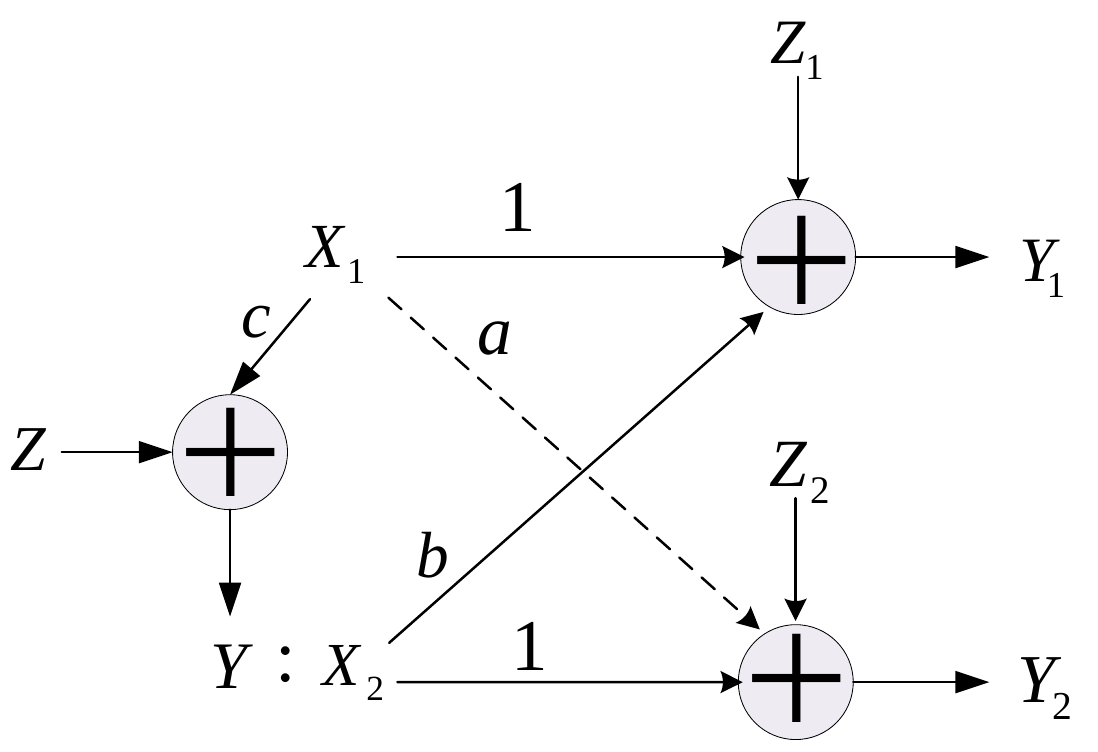}
\caption{The standard full-duplex Gaussian cognitive relay channel.}
\label{fig:FD_Gaussian_CRC}
\end{figure}
\section{Full-duplex Gaussian CRC rate regions}\label{sec:FD_Gaussian}
\subsection{Full-duplex Gaussian CRC model}
In this section, we analyze the standard full-duplex Gaussian cognitive relay channel model as follows.
\begin{align}\label{eq:FD_Gaussian_Y}
    Y_{1} &= X_1 + bX_2 + Z_1 \nonumber \\
    Y_2   &= aX_1 + X_2 + Z_2  \nonumber  \\
    Y     &= cX_1 + Z,
\end{align}
where $Z_1$, $Z_2$, $Z \sim \mathcal{N}(0,1)$ are independent Gaussian noises.
Assume that the transmit signals $X_1$ and $X_2$ are subject to power constraints $P_1$ and $P_2$, respectively.

The standard Gaussian CRC is shown in Figure \ref{fig:FD_Gaussian_CRC}. If the original channel is not in this standard form,
we can always transform it into the standard form using a procedure similar to the interference channel \cite{Carleial78}.
\subsection{Signaling and rates for full-duplex PDF-binning}\label{sec:signal_FD_PDF_Binning}
In the Gaussian channel, the signals $T_{10}$, $U_{10}$, $U_2$, $X_1$ and  $X_2$ of the PDF-binning scheme in Section \ref{sec:FD_DMC_Bin} can be represented as follows.
\begin{align}
    T_{10} &= \alpha S_{10}'(w_{10}'),\nonumber\\
    U_{10} &= \alpha S_{10}'(w_{10}') + \beta  S_{10}(w_{10}), \nonumber \\
    X_1 &= \alpha S_{10}'(w_{10}') + \beta  S_{10}(w_{10}) + \gamma  S_{11}(w_{11}), \nonumber  \\
    X_2 &= \mu\left(\rho S_{10}'(w_{10}')+\sqrt{1-\rho^2}S_{22}\right),\nonumber \\
    U_2 &= X_2 + \lambda S_{10}' = (\mu\rho+\lambda)S_{10}' +\mu\sqrt{1-\rho^2}S_{22} .\label{eq:FD_Gaussian_PDF_Bin_signaling}
\end{align}
where $S_{10}'$, $S_{10}$, $S_{11}$ and $S_{22}$ are independent $\mathcal{N}(0,1)$ random variables to encode $w_{10}'$, $w_{10}$, $w_{11}$ and $w_{2}$ respectively.
$U_2$ is the auxiliary random variable for binning that encodes $w_2$.
$X_1$ and $X_2$ are the transmit signals of $S_1$ and $S_2$.
The parameters $\alpha$, $\beta$, $\gamma$, $\mu$ are power allocation factors satisfying the power constraints
\begin{align}\label{eq:FD_Gaussian_PDF_power_constraints}
    \alpha^2 + \beta^2 + \gamma^2 \leq P_1, \nonumber\\
    \mu ^2 \leq P_2,
\end{align}
where $P_1$ and $P_2$ are transmit power constraints of $S_1$ and $S_2$.

An important feature of the signaling design in \eqref{eq:FD_Gaussian_PDF_Bin_signaling} is $\rho$ ($-1 \leq \rho \leq 1$), the correlation factor
between the transmit signal ($X_2$) and the state ($S_{10}'$) at $S_2$.
In traditional dirty paper coding, the transmit signal and the state are independent. Here we introduce correlation between them, which includes
dirty paper coding as a special case when $\rho=0$.
This correlation allows both signal forwarding and traditional binning at the same time.
$\lambda$ is the partial decode-forward binning parameter which will be optimized later.

Substitute $X_1$, $X_2$ into $Y_1$, $Y_2$ and $Y$ in \eqref{eq:FD_Gaussian_Y}, we get
\begin{align}
    Y_{1} &= (\alpha+b\mu\rho) S_{10}' + \beta  S_{10} + \gamma  S_{11} + b \mu\sqrt{1-\rho^2}  S_{22} + Z_1,  \nonumber   \\
    Y_{2} &= (a\alpha+\mu\rho) S_{10}' + a\beta  S_{10} + a \gamma  S_{11} +  \mu\sqrt{1-\rho^2}  S_{22} + Z_2,  \nonumber  \\
    Y     &= c \alpha  S_{10}' + c \beta  S_{10} + c \gamma  S_{11}  + Z.
\end{align}

\begin{cor}\label{cor:FD_Gaussian_PDF_Bin}
The achievable rate region for the full-duplex Gaussian-CRC using the PDF-binning scheme is the convex hull of all rate pairs ($R_1$, $R_2$) satisfying
\begin{align}\label{eq:FD_Gaussian_Bin}
           R_{1} &\leq  C \left(\frac  {c^2\beta ^2} {c^2\gamma ^2 + 1}\right)
                      + C \left( \frac  {\gamma ^2} {b^2\mu ^2(1-\rho^2) + 1}  \right)\nonumber\\
           R_{1} &\leq C \left( \frac  {(\alpha+b\mu\rho)^2 + \beta^2 + \gamma^2} {b^2 \mu ^2(1-\rho^2) + 1}  \nonumber\right)\\
           R_2   &\leq C\left(  \frac{\mu ^2(1-\rho^2)}{a^2\beta^2 + a^2\gamma^2 + 1}  \right)
\end{align}
where $-1 \leq \rho \leq 1$, $C(x) = \frac{1}{2}\log (1+x)$, and the power allocation factors $\alpha$, $\beta$, $\gamma$ and $\mu$ satisfy the power constraints \eqref{eq:FD_Gaussian_PDF_power_constraints}.
\end{cor}
\begin{proof}
Applying Theorem 1 with the signaling in \eqref{eq:FD_Gaussian_PDF_Bin_signaling}, we get the rate region in Corollary \ref{cor:FD_Gaussian_PDF_Bin}.
\end{proof}
\begin{rem}Maximum rates for each sender
    \begin{itemize}
        \item
        Setting $\rho=\pm 1$, $\mu =\rho \sqrt{P_2}$, we obtain the maximum rate for $R_1$ as in partial decode-forward relaying:
            \begin{align}\label{eq:FD_Gaussian_Bin_maxR1}
                R_{1}^{\max} = \max_{\substack{\alpha^2 + \beta^2 + \gamma^2 \leq P_1}} \min  \Bigg \{&C \left( \frac{c^2 \beta^2}{c^2 \gamma^2+1} \right) + C(\gamma^2), 
                                  C\left( \left(\alpha+b\sqrt{P_2}\right)^2 + \beta^2 + \gamma^2 \right) \Bigg \}.
            \end{align}
        \item
        Setting $\rho=0$, $\beta = \gamma = 0$ and $\mu=\sqrt{P_2}$, we obtain the maximum rate for $R_2$ as in dirty paper coding:
            \begin{align}\label{eq:FD_Gaussian_Bin_maxR2}
                R_{2}^{\max} = C(P_2).
            \end{align}
    \end{itemize}
\end{rem}
\subsection{Optimal binning parameter for full-duplex PDF-binning}
In this section, we derive in closed form the optimal binning parameter $\lambda$ for \eqref{eq:FD_Gaussian_PDF_Bin_signaling}
to achieve rate region \eqref{eq:FD_Gaussian_Bin}.
This optimal binning parameter is different from the optimal binning parameter in dirty paper coding,
as we introduce the correlation factor $\rho$ between the transmit signal and the state.
This correlation contains the function of message forwarding.
For example, if we set $\rho=\pm 1$, $X_2$ will only encode $w_{10}'$ without any actual binning, hence realize the function of message forwarding.
If we set $\rho=0$, PDF-binning becomes dirty paper coding without any message forwarding.
For $0 < \left|{\rho} \right| < 1$, PDF-binning has both the functions of binning and message forwarding.
Thus, PDF-binning generalizes dirty paper coding.
\begin{thm}\label{thm:FD_Gaussian_bin_lambda}
The optimal $\lambda$ for the full-duplex PDF-binning scheme is
\begin{align}\label{eq:lambda_FD_Gaussian_Bin}
    \lambda^* = \frac{ a\alpha\mu^2(1-\rho^2) - \mu\rho(a^2\beta^2+a^2\gamma^2+1) }{a^2\beta^2 + a^2\gamma^2 + \mu^2(1-\rho^2) + 1}.
\end{align}
\end{thm}
\begin{proof}
The optimal $\lambda^*$ is obtained by maximizing both rates $R_1$ and $R_2$.
In rate region \eqref{eq:FD_DMC_Bin}, through the Fourier Motzkin Elimination process, we can see that if we maximize the term $I(U_{2};Y_{2}) - I(U_{2};T_{10})$, both $R_1$ and $R_2$ are maximized simultaneously. We have
\begin{align*}
    &I(U_{2};Y_{2}) - I(U_{2};T_{10}) \\
    &= H(Y_2) - H(Y_2|U_2) - H(U_2) + H(U_2|T_{10})\\
    &= H(Y_2) + H(U_2|T_{10}) - H(U_2, Y_2).
\end{align*}
Here $\lambda$ only affects the last term $H(U_2, Y_2)$.
The covariance matrix between $U_{2}$ and $Y_{2}$ is
\begin{align}\label{eq:FD_Gaussian_Bin_cov}
\text{cov}(U_{2}, Y_{2}) =
    \begin{bmatrix}
    \text{var}(U_{2})        & \text{E}(U_{2},Y_{2})\\
    \text{E}(U_{2},Y_{2})    & \text{var}(Y_{2})
    \end{bmatrix},
\end{align}
where
\begin{align*}
    \text{var}(U_{2}) &= \mu  ^2 + \lambda ^2 + 2\mu\rho\lambda, \\
    \text{E}(U_{2},Y_{2})&=(\mu\rho+\lambda)(a\alpha+\mu\rho) + \mu^2(1-\rho^2),\\
    \text{var}(Y_{2}) &= (a\alpha+\mu\rho)^2+ a^2\beta^2 + a^2\gamma^2 + \mu^2(1-\rho^2) + 1.
\end{align*}
Minimizing the determinant of the covariance matrix
in \eqref{eq:FD_Gaussian_Bin_cov}, we obtain the optimal $\lambda^*$ in \eqref{eq:lambda_FD_Gaussian_Bin}.
\end{proof}
\begin{figure}[t]
\centering
\includegraphics[scale=0.45]{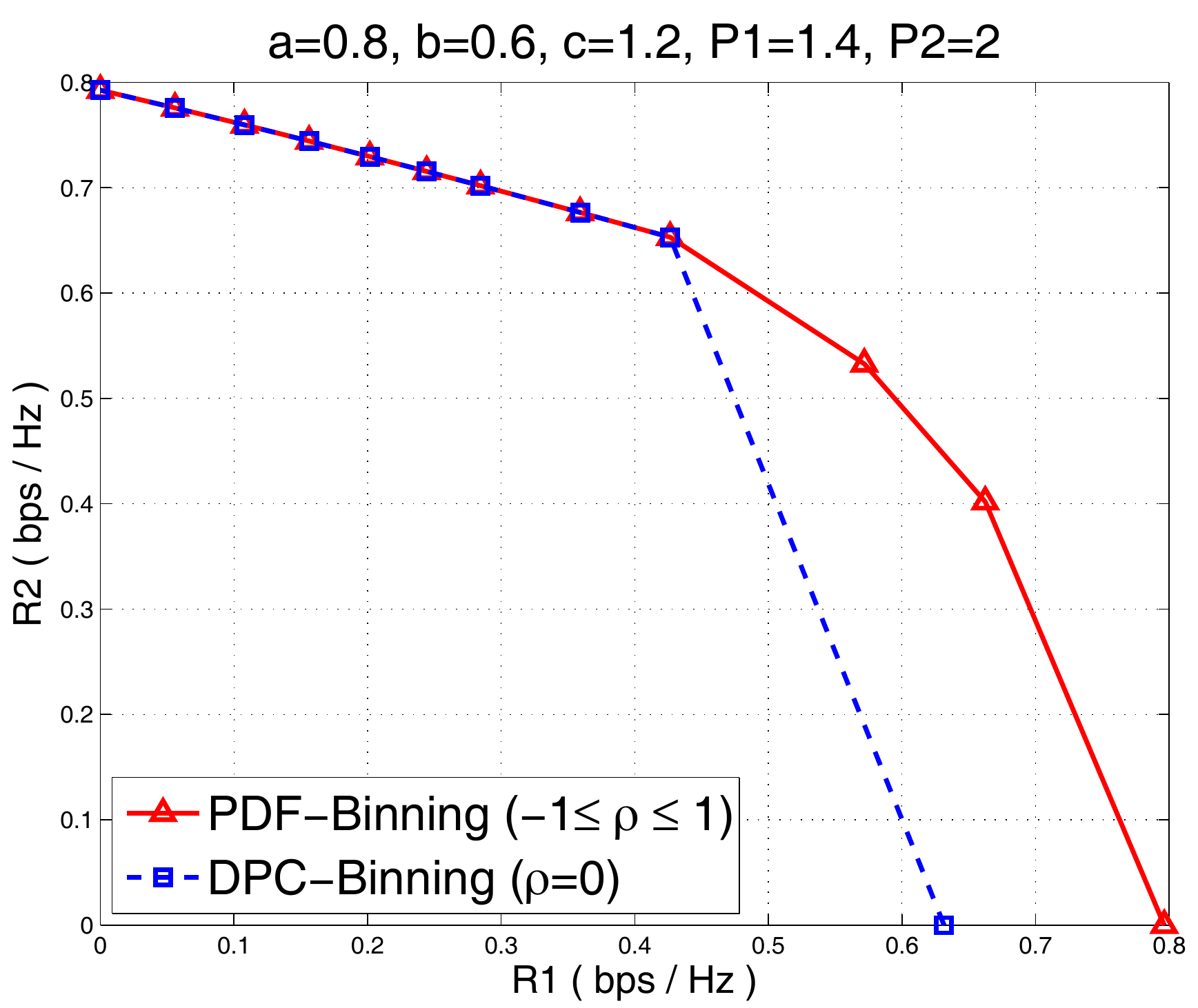}
\caption{Effect of the binning correlation factor $\rho$.}
\label{fig:FD_PDFBin_effect_of_rho}
\end{figure}
\begin{rem}Effect of $\rho$:
\begin{itemize}
\item
If $\rho=0$, $\lambda^*$ becomes the optimal $\lambda$ for traditional dirty paper coding \cite{Costa83TIT}, which achieves the maximum rate for $R_2$
as in \eqref{eq:FD_Gaussian_Bin_maxR2}.
\item
If $\rho=\pm 1$, $\lambda^*$ differs from the $\lambda$ in traditional dirty paper coding and achieves the maximum rate for $R_1$ as in \eqref{eq:FD_Gaussian_Bin_maxR1}.
\item
The effect of $\rho$ can be seen in Figure \ref{fig:FD_PDFBin_effect_of_rho}.
The dashed line represents the resulting rate region using only DPC-binning ($\rho=0$),
 while the solid line represents the region for PDF-binning when we adapt $\rho \in [-1, 1]$. Figure \ref{fig:FD_PDFBin_effect_of_rho} illustrates that the correlation factor $\rho$ can enlarge the rate region.
\end{itemize}
\end{rem}
\subsection{Signaling and rates for full-duplex Han-Kobayashi PDF-binning}\label{sec:FD_Gaussian_combined}
In the Gaussian channel, input signals for the HK-PDF-binning scheme in Section \ref{sec:FD_DMC_Combined} can be represented as
\begin{align}\label{eq:FD_Gaussian_Combined_signal}
    T_{10} &= \alpha S_{10}'(w_{10}'), \nonumber \\
    U_{10} &= \alpha S_{10}'(w_{10}') + \beta S_{10}(w_{10}), \nonumber \\
    U_{11} &= \gamma S_{11}(w_{11}),\nonumber \\
    X_1 &= \alpha S_{10}'(w_{10}') + \beta S_{10}(w_{10}) + \gamma S_{11}(w_{11}) + \delta S_{12}(w_{12}), \nonumber  \\
    U_{21} &=  \theta S_{21}(w_{21}),  \nonumber  \\
    X_2 &= \theta S_{21}(w_{21}) + \mu\left(\rho  S_{10}'(w_{10}')   + \sqrt{1-\rho^2}S_{22}\right),\nonumber \\
    U_{22} &= X_2  + \lambda S_{10}'  
           = (\mu\rho+\lambda)S_{10}' + \theta S_{21}(w_{21}) + \mu\sqrt{1-\rho^2}S_{22},
\end{align}
where $S_{10}'$, $S_{10}$, $S_{11}$, $S_{12}$, $S_{21}$, $S_{22}$ are independent $\mathcal{N}(0,1)$ random variables to encode $w_{10}'$, $w_{10}$, $w_{11}$, $w_{12}$, $w_{21}$, $w_{22}$, respectively.
$U_{22}$ is the auxiliary random variable for binning that encodes $w_{22}$.
$X_1$ and $X_2$ are the transmit signals of $S_1$ and $S_2$.
$\rho$ is the correlation coefficient between the transmit signal and the binning state at $S_2$ ($-1\leq \rho \leq 1$).
$\lambda$ is the PDF-binning parameter.
The   parameters $\alpha$, $\beta$, $\gamma$, $\delta$, $\theta$ and $\mu$ are power allocation factors satisfying the power constraints
\begin{align}\label{eq:FD_Gaussian_combined_power_constraints}
    \alpha^2 + \beta^2 + \gamma^2 + \delta^2 &\leq P_1, \nonumber\\
    \theta ^2 + \mu^2 &\leq P_2,
\end{align}
where $P_1$ and $P_2$ are transmit power constraints of $S_1$ and $S_2$.

Substitute these variables into the Gaussian channel in \eqref{eq:FD_Gaussian_Y}, we get
\begin{align}
    Y      &= c\alpha S_{10}' + c\beta S_{10} + c\gamma S_{11} + c\delta S_{12} + Z, \nonumber   \\
     Y_1   &= (\alpha + b\mu \rho) S_{10} ' + \beta S_{10} + \gamma S_{11} + \delta S_{12}  
            + b\theta S_{21} + b\mu \sqrt{1-\rho^2}S_{22}+Z_1,\nonumber \\
    Y_2    &= (a\alpha + \mu \rho) S_{10} ' + a\beta S_{10} + a\gamma S_{11} + a\delta S_{12} 
            + \theta S_{21} + \mu \sqrt{1-\rho^2}S_{22}+Z_2.
\end{align}
\begin{cor}\label{cor:FD_Gaussian_combined}
The achievable rate region for the full-duplex Gaussian-CRC using the Han-Kobayashi PDF-binning scheme is the convex hull of all rate pairs ($R_1$, $R_2$) satisfying
\begin{align}\label{eq:FD_Gaussian_Combined}
                             R_{1} &\leq \min\{I_2 + I_5, I_6\}\nonumber \\
                             R_{2} &\leq I_{12}- I_1\nonumber \\
                     R_{1} + R_{2} &\leq \min\{I_2 + I_7,I_8 \} + I_{13}- I_1\nonumber \\
                     R_{1} + R_{2} &\leq \min\{I_2 + I_3,I_4\} + I_{14}- I_1\nonumber \\
                     R_{1} + R_{2} &\leq \min\{I_2 + I_9, I_{10} \} + I_{11} - I_1 \nonumber\\
                    2R_{1} + R_{2} &\leq \min\{ I_2 + I_3,I_4\} + \min\{I_2 + I_9, I_{10} \} + I_{13}- I_1\nonumber \\
                    R_{1} + 2R_{2} &\leq \min\{I_2 + I_7, I_8 \}+ I_{11} - I_1 + I_{14}- I_1
\end{align}
where
\begin{align*}
    I_2 &= C\left(\frac{c^2\beta^2}{c^2\gamma^2 + c^2\delta^2 + 1} \right)\\
    I_3 &= C\left(\frac{\delta^2}{b^2\mu^2 (1-\rho^2)+1} \right)\\
    I_4 &= C\left(\frac{\beta^2 + \delta^2 }{b^2\mu^2 (1-\rho^2)+1} \right)
         + C\left(\frac{(\alpha + b\mu \rho)^2 }{\beta^2 + \gamma^2 + \delta^2 + b^2\theta^2 + b^2\mu^2 (1-\rho^2)+ 1} \right)\\
    I_5 &= C\left(\frac{\gamma^2 + \delta^2}{b^2\mu^2 (1-\rho^2)+1} \right)\\
    I_6 &= C\left(\frac{\beta^2 + \gamma^2 + \delta^2}{b^2\mu^2 (1-\rho^2)+1} \right)
         + C\left(\frac{(\alpha + b\mu \rho)^2 }{\beta^2 + \gamma^2 + \delta^2 + b^2\theta^2 + b^2\mu^2 (1-\rho^2)+ 1} \right)\\
    I_7 &= C\left(\frac{\delta^2 + b^2\theta^2}{b^2\mu^2 (1-\rho^2)+1} \right)\\
    I_8 &= C\left(\frac{\beta^2 + \delta^2 + b^2\theta^2}{b^2\mu^2 (1-\rho^2)+1} \right)
         + C\left(\frac{(\alpha + b\mu \rho)^2 }{\beta^2 + \gamma^2 + \delta^2 + b^2\theta^2 + b^2\mu^2 (1-\rho^2)+ 1} \right)\\
    I_9 &= C\left(\frac{\gamma^2 + \delta^2 + b^2\theta^2}{b^2\mu^2 (1-\rho^2)+1} \right)\\
    I_{10} &= C\left(\frac{(\alpha + b\mu \rho)^2 + \beta^2 + \gamma^2 + \delta^2 + b^2\theta^2}{b^2\mu^2 (1-\rho^2)+1} \right)\\
    I_{11}-I_{1}&=C \left( \frac {\mu^2(1-\rho^2) } {a^2\beta^2+a^2\delta^2 + 1} \right)\\
    I_{12} - I_1 &= C \left( \frac {\mu^2(1-\rho^2) } {a^2\beta^2+a^2\delta^2 + 1} \right) 
          + C\left( \frac{\theta^2 }{(a\alpha + \mu \rho)^2 + a^2\beta^2 + a^2\delta^2 + \mu^2 (1-\rho^2)+1} \right)
\end{align*}
\begin{align*}
    I_{13} - I_1 &= C \left( \frac {\mu^2(1-\rho^2) } {a^2\beta^2+a^2\delta^2 + 1} \right)
          + C\left( \frac{a^2\gamma^2}{(a\alpha + \mu \rho)^2 + a^2\beta^2 + a^2\delta^2 + \mu^2 (1-\rho^2)+1} \right)\\
    I_{14} - I_1 &= C \left( \frac {\mu^2(1-\rho^2) } {a^2\beta^2+a^2\delta^2 + 1} \right)
          + C\left( \frac{a^2\gamma^2+\theta^2}{(a\alpha + \mu \rho)^2 + a^2\beta^2 + a^2\delta^2 + \mu^2 (1-\rho^2)+1} \right)
\end{align*}
and $\alpha$, $\beta$, $\gamma$, $\delta$, $\theta$ and $\mu$ are power allocation factors satisfying the power constraints
\eqref{eq:FD_Gaussian_combined_power_constraints} and $-1\leq \rho \leq 1$.
\end{cor}
\begin{proof}
Applying Theorem \ref{thm:FD_DMC_combined} with the signaling in \eqref{eq:FD_Gaussian_Combined_signal}, we obtain the rate region in Corollary \ref{cor:FD_Gaussian_combined}.
\end{proof}

Note that rate region \eqref{eq:FD_Gaussian_Combined} includes both the Han-Kobayashi rate region and the PDF-binning region
in \eqref{eq:FD_Gaussian_Bin}.
Furthermore, the maximum rates for user 1 and user 2 are the same as in \eqref{eq:FD_Gaussian_Bin_maxR1} and \eqref{eq:FD_Gaussian_Bin_maxR2}.
\subsection{Optimal binning parameter for full-duplex Han-Kobayashi PDF-binning}
\begin{cor}\label{cor:FD_Gaussian_optimalLambda}
The optimal $\lambda^*$ for the full-duplex HK-PDF-binning scheme is
\begin{align}\label{eq:lambda_FD_Gaussian_Combined}
    \lambda^* &= \frac{a\alpha\mu^2(1-\rho^2)-\mu\rho(a^2\beta^2+a^2\delta^2+1)}{a^2\beta^2+a^2\delta^2+\mu^2(1-\rho^2)+1}
\end{align}
\end{cor}
\begin{proof}
     $\lambda^*$ is obtained by maximizing the term $I_{11}-I_{1}$ in \eqref{eq:FD_DMC_combined}.
     See Appendix \ref{proof:FD_Gaussian_Combined_lambda} for details.
\end{proof}

     Note that the optimal $\lambda^*$ in \eqref{eq:lambda_FD_Gaussian_Combined} contains both the optimal $\lambda^*$ for PDF-binning in \eqref{eq:lambda_FD_Gaussian_Bin} and
     the optimal $\lambda$ for DPC binning \cite{Costa83TIT} as special cases.

\subsection{Numerical examples}\label{sec:FD_Outer_compare}
In this section, we provide numerical comparison among the proposed PDF-binning and HK-PDF-binning schemes, the original Han-Kobayashi scheme,
and an outer bound as discussed below.
\subsubsection{Outer bounds for the CRC capacity}
We obtain a simple outer bound for the CRC capacity by combining the capacity for the (non-causal) CIC and the ourter bound for interference channel with user cooperation (IC-UC) \cite{TandonTIT2011}.
Where the CIC capacity result is not available, we use the MISO broadcast capacity.
\begin{align*}
\text{CRC capacity} &\subset \text{CIC capacity} \bigcap \text{IC-UC outer bound} \\
                    &\subset \text{MISO BC capacity} \bigcap \text{IC-UC outer bound}.
\end{align*}

\emph{a)} {Capacity of the CIC as an outer bound:}
The capacity of the ideal CIC (with non-causal knowledge of $S_1$'s message at $S_2$) is an outer bound to the CRC rate region.
The CIC capacity is known in the cases of (i) weak interference \cite{Wu2007TIT,Jovicic2009TIT} (ii) very strong interference \cite{Maric2007TIT} (iii) the primary-decode-cognitive region \cite{Rini2011TIT_Gaussian}.
For strong interference, we can also use the outer bound to the CIC capacity in \cite{MaricEuroTele2008} as an outer bound to the CRC.

\emph{b)} {IC-UC outer bound:}
Tandon and Ulukus \cite{TandonTIT2011} obtain an outer bound for the MAC with generalized feedback based on dependence balance, which is first proposed by Hekstra and Willems \cite{HekstraTIT1989} to study outer bounds for the single-output two-way channels. The basic idea of dependence balance is that no more information can be consumed than produced. Tandon and Ulukus apply this idea to obtain a new outer bound for IC-UC.
It is shown that this dependence-balance-based outer bound is strictly tighter than the cutset bound (see Section V of \cite{TandonTIT2011}).
Thus, this bound can be used instead of the relay channel (RC) cutset bound for $R_1$.

\emph{c)} {Gaussian Vector Broadcast Outer Bound:}
Consider a $2 \times 1$ MISO broadcast system as
\begin{align}
    Y_1 = [1 \quad b] X + Z_1, \nonumber\\
    Y_2 = [a \quad 1] X + Z_2,
\end{align}
where $a$, $b$ are the channel gains,
$Z_1$ and $Z_2$ are white Gaussian noises with identity covariance.
The vector codeword $X$ consists of two independent parts:
\begin{align*}
    X = U+V,
\end{align*}
where
$X=\left(
                  \begin{array}{c}
                    X_1 \\
                    X_2 \\
                  \end{array}
                \right)
$,
$U=\left(
                  \begin{array}{c}
                    U_1 \\
                    V_1 \\
                  \end{array}
                \right)
$,
$V=\left(
                  \begin{array}{c}
                    U_2 \\
                    V_2 \\
                  \end{array}
                \right)
$, and $U_1$, $V_1$, $U_2$, $V_2$ are zero-mean Gaussian codewords with covariances:
\begin{align*}
K_U =
    \begin{bmatrix}
              \alpha^2                  & \rho_1 \alpha \beta \\
              \rho_1 \alpha \beta  & \beta^2
    \end{bmatrix}; \quad
K_V =
\begin{bmatrix}
          \gamma^2                  & \rho_2 \gamma \delta \\
          \rho_2 \gamma \delta  & \delta^2
\end{bmatrix},
\end{align*}
in which the power allocation factors satisfy
\begin{align}
    \alpha^2 + \beta^2 \leq P_1, \quad \gamma^2 + \delta^2 \leq P_2,
\end{align}
and the input correlation factors $\rho_1, \rho_2\in [-1,1]$.

The Gaussian vector broadcast capacity region is the convex closure of $R_{o1} \bigcup R_{o2}$ \cite{Gamal10notes},
where $R_{o1}$ is the region
\begin{align}
                R_1 &\leq 
                                         C\left(
                                             \frac{\alpha^2 + 2b \rho_1 \alpha \beta + b^2 \beta^2 }
                                             {\gamma^2 + 2b \rho_2 \gamma \delta + b^2 \delta^2+1}
                                             \right)\nonumber\\
                R_2 &\leq 
                                     C\left(a^2\gamma^2 + 2a \rho_2 \gamma \delta +  \delta^2 \right)
\end{align}

and $R_{o2}$ is the region
\begin{align}
                R_1 &\leq 
                             C\left(\alpha^2 + 2b \rho_1 \alpha \beta +  b^2 \beta^2 \right) \nonumber\\
                R_2 &\leq 
                                         C\left(
                                             \frac{a^2 \gamma^2 + 2a \rho_2 \gamma \delta + \delta^2}
                                             {a^2 \alpha^2 + 2a \rho_1 \alpha \beta +  \beta^2+1}
                                             \right).
\end{align}

\subsubsection{Numerical comparison}
\begin{figure}[t]
\centering
\includegraphics[scale=0.45]{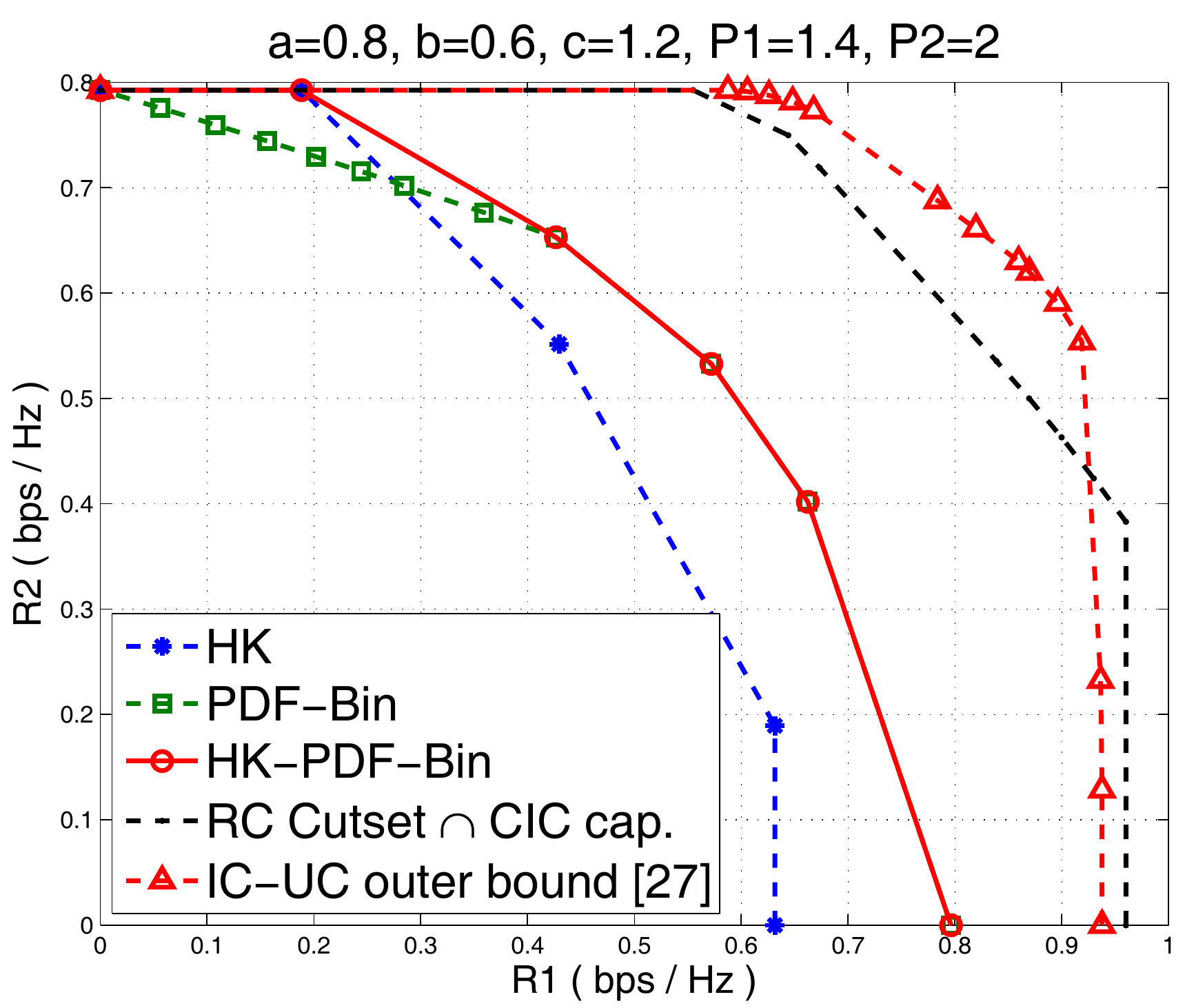}
\caption{Rate regions for full-duplex schemes in the Gaussian cognitive relay channel. }
\label{fig:FD_PDFBin_HKPDFBin_comparison}
\end{figure}

Figure \ref{fig:FD_PDFBin_HKPDFBin_comparison} shows the comparison in the full-duplex mode among the Han-Kobayashi scheme,
PDF-binning, HK-PDF-binning, and the outer bound.
We can see that the proposed HD-PDF-binning scheme contains both the Han-Kobayashi and the PDF-binning rate regions, as analyzed in Remark \ref{rem:FD_DMC_Combined_Inclusions}.
Note that the outer bound is the intersection of the two bounds drawn and is loose as this bound is not achievable.
However, we observe that as $b$ decreases, the HK-PDF-binning rate region becomes closer to the outer bound.

%

\section{Half-Duplex Coding schemes}
In this section, we adapt the two full-duplex schemes to the half-duplex mode.
The half-duplex schemes are also based on rate splitting, superposition encoding, partial
decode-forward binning and Han-Kobayashi coding.
There are several differences between the half- and full-duplex cases. First, under the half-duplex constraint, no node can
both transmit and receive at the same time, thus leading us to divide each transmission block into two phases.
In the first phase, $S_1$ sends a message to $S_2$, $T_1$ and $T_2$, while $S_2$ only receives but sends no messages.
In the second phase, both $S_1$  and $S_2$ send messages concurrently.
Second, $S_1$ sends different message parts in different phases.
Specifically, $S_1$ only sends one part of its message to other nodes in the first phase, but will send all message parts in the second phase.
Third, there is no block Markovity in the encoding since the superposition coding can be done between 2 phases of the same block instead of between 2 consecutive blocks.
Finally, both $T_1$ and $T_2$ apply joint decoding only at the end of the second phase to make use of the received signals in both phases.

\subsection{Half-duplex partial decode-forward binning scheme}
\begin{figure*}[t]
\centering
\includegraphics[scale=0.7]{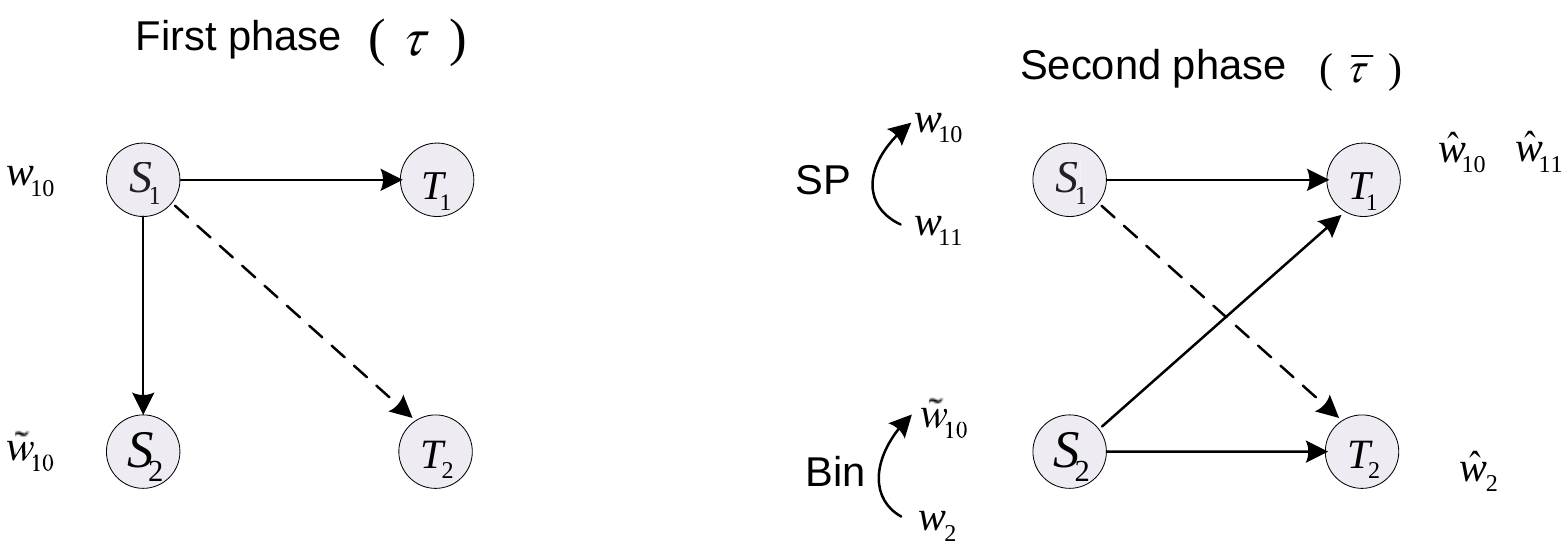}
\caption{Coding structure for the half-duplex CRC based on partial decode-forward binning.}
\label{fig:HD_DM_CRC_PDF_Bin}
\end{figure*}
The coding structure for the half-duplex PDF-binning scheme is shown in Figure \ref{fig:HD_DM_CRC_PDF_Bin}.
This scheme uses superposition encoding at the first sender, and partial decode-forward relaying and binning at the second sender.
The first sender $S_1$ splits its message into two parts $(w_{10}, w_{11})$, corresponding to the forwarding and private parts.
In the first phase, $S_1$ sends a codeword $X_{11}^{\tau n}$ containing the message part $w_{10}$;
$S_2$ sends no information but only listens. At the end of the first phase, $S_2$ decodes $w_{10}$ from $S_1$ .
Note that neither $T_1$ nor $T_2$ decodes during this phase.
In the second phase, $S_1$ sends a codeword $X_{12}^{\tau n}$ containing both parts $(w_{10}, w_{11})$, in which $w_{11}$ is superimposed on $w_{10}$.
$S_2$ now sends both $w_2$ and $w_{10}$ and uses Gelfand-Pinsker binning technique to bin against the codeword $X_{11}^{n}(w_{10})$ decoded from $S_1$ in the first phase.
At the destinations, $T_1$ uses joint decoding to decode $(w_{10},w_{11})$ from the signals received in both phases;
$T_2$ decodes $w_{2}$ using the received signal in the second phase.

Specifically, at the end of the first phase, $S_2$ searches for a unique $\hat{w}_{10}$ such that
\begin{align*}
    (x_{11}^{ \tau n}{(\hat{w}_{10})}, \mathbf{y})
    \in A_\epsilon^{(\tau n)}(P_{X_{11} Y}),
\end{align*}
where $\mathbf{y}$ is the received signal vector at $S_2$ in the first phase.
It then performs binning by looking for a $v_2$ such that
\begin{align*}
    (x_{11}^{ \bar{\tau} n}{(\hat{w}_{10})}, u_{2}^{ \bar{\tau} n}{(w_{2},v_{2})})
    \in A_\epsilon^{(\bar{\tau} n)}(P_{X_{1} U_{2}}),
\end{align*}
and sends $x_{22}^{ \bar{\tau} n}{(\mathbf{x}_{11}, \mathbf{u}_{2})}$ as a function of $x_{11}^{ \bar{\tau} n}$ and $u_{2}^{ \bar{\tau} n}$ in the second phase.

At the end of the second phase, $T_1$ searches for a unique ($\hat{w}_{10},\hat{w}_{11}$) such that
\begin{align*}
    ( x_{11}^{\bar{\tau}n}{(\hat{w}_{10})}, x_{12}^{\bar{\tau}n}(\hat{w}_{11}|\hat{w}_{10}), \mathbf{y_{12}})&\in A_\epsilon^{(\bar{\tau}n)}(P_{X_{11} X_{12} Y_{12}})&\\
     \text{and}  \quad  (x_{11}^{ \tau n}{(\hat{w}_{10})},  \mathbf{y_{11}}) &\in A_\epsilon^{(\tau n)}(P_{X_{11} Y_{11}}),&
\end{align*}
where $\mathbf{y_{11}}$ and $\mathbf{y_{12}}$ indicate the received vectors at $T_1$ during the first and second phases, respectively.
$T_2$ treats the codeword $X_{11}^n$ as the state and decodes $w_{2}$.
It searches for a unique $\hat{w}_2$ for some $\hat{v}_{2} $  such that
\begin{align*}
    (u_2^{\bar{\tau}n}(\hat{w}_{2} ,\hat{v}_{2}), \mathbf{y_{22}}) \in   A_\epsilon^{(\bar{\tau}n)}(P_{U_2 Y_{22}}),
\end{align*}
where $\mathbf{y_{22}}$ is the received vector at $T_2$ in the second phase.
\begin{thm}
\label{thm:HD_DMC_Bin}
The convex hull of the following rate region is achievable for the half-duplex cognitive relay channel using PDF-binning:
\begin{align}\label{eq:DMC_PDF_Bin}
\bigcup_{\substack{P_{3}}}
\left\{\begin{array}{ll}
             R_{1} &\leq  \tau  I(X_{11};Y) + \bar{\tau} I(X_{12};Y_{12}|X_{11})\\
             R_{1} &\leq  \tau  I(X_{11}; Y_{11}) + \bar{\tau} I(X_{11}, X_{12};Y_{12})\\
             R_2 &\leq \bar{\tau} I(U_2;Y_{22})- \bar{\tau} I( U_2; X_{11})
\end{array}\right.
\end{align}
where
\begin{align*}
    P_{3} = &p(x_{11})p(x_{12}|x_{11})p(u_2|x_{11})p(x_{22}|x_{11},u_2)
            p_c(y_{11},y_{21},y,y_{12},y_{22}|x_{11},x_{12},x_{22}),
\end{align*}
and $p_c$ is given in \eqref{eq:p_c}, $\bar{\tau} = 1 - \tau, 0\leq \tau \leq 1$.
\end{thm}
\begin{proof}
See Appendix \ref{proof:HD_DMC_Bin} for the detailed proof.
\end{proof}

\begin{rem}The maximum rate for each user.
    \begin{itemize}
    \item
    The first user $S_1$ achieves the maximum rate of half-duplex partial decode-forward relaying if we set $U_2=\emptyset$.
    \begin{align}\label{eq:HD_DM_Bin_maxR1}
        R_{1}^{\max} =  \max_{\substack{0\leq \tau \leq 1 \\p(x_{11},x_{12})}}\min\{&\tau  I(X_{11};Y) + \bar{\tau} I(X_{12};Y_{12}|X_{11}), 
                          \tau  I(X_{11}; Y_{11}) + \bar{\tau} I(X_{11}, X_{12};Y_{12})\}.
    \end{align}
    This half-duplex $R_{1}^{\max}$ is slightly smaller than in the full-duplex case of \eqref{eq:FD_DMC_Bin_maxR1}.
    \item
    The second user $S_2$ achieves the maximum rate of Gelfand-Pinsker's binning if we set $\tau=0$, $X_{12}=X_{11}$.
    \begin{align}\label{eq:HD_DM_Bin_maxR2}
        R_{2}^{\max} = \max_{\substack{p(x_{11},u_2)p(x_{22}|x_{11},u_2)}} \{I(U_2;Y_{22})- I( U_2; X_{11})\}.
    \end{align}
    This half-duplex $R_{2}^{\max}$ is the same as in the full-duplex case of \eqref{eq:FD_DMC_Bin_maxR2}.
    Even though this equality seems somewhat surprising, it is indeed the case in the limit of $\tau \rightarrow 0$, given
    that user 1 sends just enough information for $S_2$ to be able to decode completely in the first phase and then bin against it in the second phase.
    At $\tau = 0$ and $X_{12} = X_{11} = \emptyset$, $S_2$ can achieve the interference-free rate.
    \end{itemize}
\end{rem}
\subsection{Half-duplex Han-Kobayashi PDF-binning scheme}\label{sec:HD_DMC_Combined}

\begin{figure*}[t]
\centering
\includegraphics[scale=0.7]{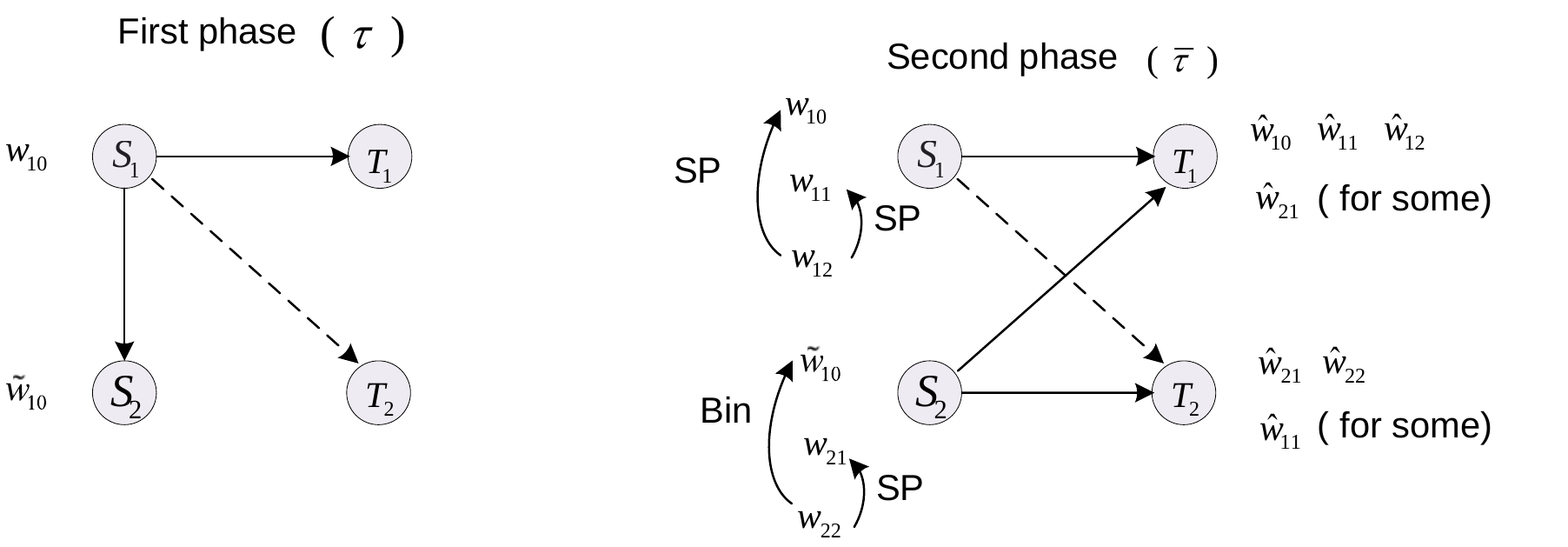}
\caption{Coding structure for the half-duplex CRC based on Han-Kobayashi partial decode-forward binning.}
\label{fig:HD_DM_CRC_Combined}
\end{figure*}

The first half-duplex coding scheme utilizes PDF-binning at the second sender and achieves the maximum possible rates for both user $1$ and user $2$.
But it does not include the Han-Kobayashi scheme for the interference channel.
In this section, we extend this scheme to combine with the Han-Kobayashi scheme by further splitting the messages in the second phase.

The coding structure for half-duplex HK-PDF-binning is shown in Figure \ref{fig:HD_DM_CRC_Combined}.
The encoding and decoding procedure in the first phase is the same as that of half-duplex PDF-binning.
The major difference is in the second phase.
Message $w_1$ of the first sender $S_1$ is split into three parts ($w_{10},w_{11},w_{12}$), corresponding to the forwarding, public and private parts.
Message $w_2$ is split into 2 parts ($w_{21}, w_{22}$), corresponding to the public and private parts.
We generate independent codewords for messages $w_{10}$ and $w_{11}$ and superimpose $w_{12}$ on both of them.
In the first phase, $S_1$ sends a codeword containing $w_{10}$, while $S_2$ does not send any message.
At the end of the first phase, $S_2$ decode $\tilde{w}_{10}$ using the received signal vector $\mathbf{y}$ and then bins its private part $w_{22}$ against the decoded message $\tilde{w}_{10}$, conditionally on knowing the public part $w_{21}$.
In the second phase, $S_1$ sends a codeword containing ($w_{10}, w_{11}, w_{12}$) while $S_2$ sends the binned signal containing ($w_{10}, w_{21}, w_{22}$).
At the end of the second phase, $T_1$ uses joint decoding across both phases and searches for a unique triple ($\hat{w}_{10},\hat{w}_{11},\hat{w}_{12}$) for some $w_{21}$.
$T_2$ also uses joint decoding based on the received signal in the second phase and searches for a unique pair $(\hat{w}_{21},\hat{w}_{22})$
for some $\hat{w}_{11}$.

Specifically, in the first phase, $S_1$ sends $x_{11}^{ \tau n}(w_{10})$; $S_2$ does not transmit.
In the second phase, $S_1$ sends $x_{12}^{\bar{\tau}n}(w_{12}|w_{10},w_{11})$;
      $S_2$ searches for some $v_{22}$ such that
            \begin{align}\label{eq:HD_DMC_Bin_binStep}
                (x_{11}^{\bar{\tau}n}{(w_{10})},&u_{21}^{\bar{\tau}n}( w_{21}), u_{22}^{\bar{\tau}n}( w_{22}, v_{22}|w_{21})  )
                \in A_\epsilon^{(\bar{\tau}n)}(P_{X_{11} U_{22}|U_{21}}),
            \end{align}
and then sends $x_{2}^{ \bar{\tau} n}(w_{10},w_{21},w_{22},v_{22})$.

For decoding, at the end of the first phase, $S_2$ searches for a unique $\hat{w}_{10}$ such that
\begin{align}
    (x_{11}^{ \tau n}{(\hat{w}_{10})}, \mathbf{y})
    \in A_\epsilon^{(\tau n)}(P_{X_{11} Y}).
\end{align}
At the end of the second phase, $T_1$ searches for a unique ($\hat{w}_{10},\hat{w}_{11},\hat{w}_{12}$) for some $\hat{w}_{21}$ such that
\begin{align}\label{eq:HD_DMC_Bin_jointDec1}
    ( x_{11}^{\bar{\tau}n}{(\hat{w}_{10})},u_{11}^{\bar{\tau}n}{(\hat{w}_{11})}, x_{12}^{\bar{\tau}n}(\hat{w}_{12}|\hat{w}_{10},\hat{w}_{11}),
    u_{21}^{\bar{\tau}n}{(\hat{w}_{21})}, \mathbf{y_{12}})&\in A_\epsilon^{(\bar{\tau}n)}(P_{X_{11} U_{11} X_{12} U_{21} Y_{12}})\nonumber\\
     \text{and} \quad x_{11}^{ \tau n}{(\hat{w}_{10})}, \mathbf{y_{11}}) &\in A_\epsilon^{(\tau n)}(P_{ X_{11} Y_{11}}).
\end{align}
$T_2$ searches for a unique $(\hat{w}_{21},\hat{w}_{22})$ for some $(\hat{w}_{11},\hat{v}_{22})$  such that
\begin{align}\label{eq:HD_DMC_Bin_jointDec2}
    (u_{11}^{\bar{\tau}n}(\hat{w}_{11} ),&u_{21}^{\bar{\tau}n}(\hat{w}_{21} ), u_{22}^{\bar{\tau}n}(\hat{w}_{22},\hat{v}_{22}|\hat{w}_{21} ),\mathbf{y_{22}}) 
    \in   A_\epsilon^{(\bar{\tau}n)}(P_{U_{11} U_{21} U_{22} Y_{22}}).
\end{align}

Note that similar to the full-duplex scheme in Section \ref{sec:FD_DMC_Combined}, we use conditional binning in step \eqref{eq:HD_DMC_Bin_binStep},
and joint decoding at both destinations in steps \eqref{eq:HD_DMC_Bin_jointDec1} and \eqref{eq:HD_DMC_Bin_jointDec2} (see Remark \ref{rem:FD_DMC_Combined_features}).
\begin{thm}
\label{thm:HD_DMC_combined}
The convex hull of the following rate region is achievable for the half-duplex cognitive relay channel using the HK-PDF-binning scheme:
\begin{align}\label{eq:HD_DMC_combined}
\bigcup_{\substack{P_{4}}}
\left\{\begin{array}{ll}
                             R_{1} &\leq \min\{I_2 + I_5, I_6\}\\
                             R_{2} &\leq I_{12}- I_1 \\
                     R_{1} + R_{2} &\leq \min\{I_2 + I_7,I_8 \} + I_{13}- I_1 \\
                     R_{1} + R_{2} &\leq \min\{I_2 + I_3,I_4\} + I_{14}- I_1 \\
                     R_{1} + R_{2} &\leq \min\{I_2 + I_9, I_{10} \} + I_{11} - I_1 \\
                    2R_{1} + R_{2} &\leq \min\{ I_2 + I_3,I_4\} 
                    + \min\{I_2 + I_9, I_{10} \} + I_{13}- I_1\\
                    R_{1} + 2R_{2} &\leq \min\{I_2 + I_7, I_8 \} + I_{11} - I_1
                                     + I_{14}- I_1
\end{array}\right.
\end{align}

where
\begin{align}
    P_{4} = &p(x_{11})p(u_{11})p(x_{12}|u_{11},x_{11})p(u_{21})p(u_{22}|u_{21},x_{11})\nonumber\\
            &p(x_{22}|x_{11},u_{21},u_{22})
            p_c(y_{11},y_{21},y,y_{12},y_{22}|x_{11},x_{12},x_{22}),
\end{align}
with $p_c$ as given in \eqref{eq:p_c} and
\begin{align}
       I_1 &= \bar{\tau} I( U_{22}; X_{11}|U_{21}) \nonumber \\
       I_2 &= \tau  I(X_{11};Y) \nonumber \\
       I_3 &= \bar{\tau} I(U_{12};Y_{12}|X_{11},U_{11},U_{21})  \nonumber \\
       I_4 &= \tau I(X_{11};Y_{11}) + \bar{\tau} I(X_{11},X_{12};Y_{12}|U_{11},U_{21})   \nonumber \\
       I_5 &= \bar{\tau} I(U_{11},X_{12};Y_{12}|X_{11},U_{21}) \nonumber \\
       I_6 &= \tau I(X_{11};Y_{11}) + \bar{\tau} I(X_{11},U_{11},X_{12};Y_{12}|U_{21})  \nonumber \\
       I_7 &= \bar{\tau} I(X_{12},U_{21};Y_{12}|X_{11},U_{11})\nonumber\\
       I_8 &= \tau I(X_{11};Y_{11}) + \bar{\tau} I(X_{11},X_{12},U_{21};Y_{12}|U_{11})  \nonumber
\end{align}
\begin{align}
       I_9 &= \bar{\tau} I(U_{11},X_{12},U_{21};Y_{12}|X_{11})  \nonumber\\
       I_{10} &= \tau I(X_{11};Y_{11}) + \bar{\tau} I(X_{11},U_{11},X_{12},U_{21};Y_{12}) \nonumber\\
       I_{11} &= \bar{\tau} I(U_{22};Y_{22}|U_{21},U_{11})  \nonumber\\
       I_{12} &= \bar{\tau} I(U_{21},U_{22};Y_{22}|U_{11})  \nonumber\\
       I_{13} &= \bar{\tau} I(U_{11},U_{22};Y_{22}|U_{21}) \nonumber\\
       I_{14} &= \bar{\tau} I(U_{11},U_{21},U_{22};Y_{22})
\end{align}
\end{thm}
where $\bar{\tau} = 1 - \tau, 0\leq \tau \leq 1$.
\begin{proof}
See Appendix \ref{proof:HD_DMC_Combined} for the details.
\end{proof}
\begin{rem}Inclusion of half-duplex PDF-binning and Han-Kobayashi schemes.
    \begin{itemize}
    \item
    The half-duplex HK-PDF-binning scheme becomes half-duplex PDF-binning if $U_{11}=U_{21}=\emptyset$.
    \item
    The half-duplex HK-PDF-binning scheme becomes the Han-Kobayashi scheme if $\tau=0$, $X_{11}=\emptyset$ and $X_{22} = U_{22}$.
    \item
    The maximum rates for $S_1$ and $S_2$ are the same as in \eqref{eq:HD_DM_Bin_maxR1} and \eqref{eq:HD_DM_Bin_maxR2}.
    \end{itemize}
\end{rem}

\begin{figure*}[t]
\centering
\includegraphics[scale=0.60]{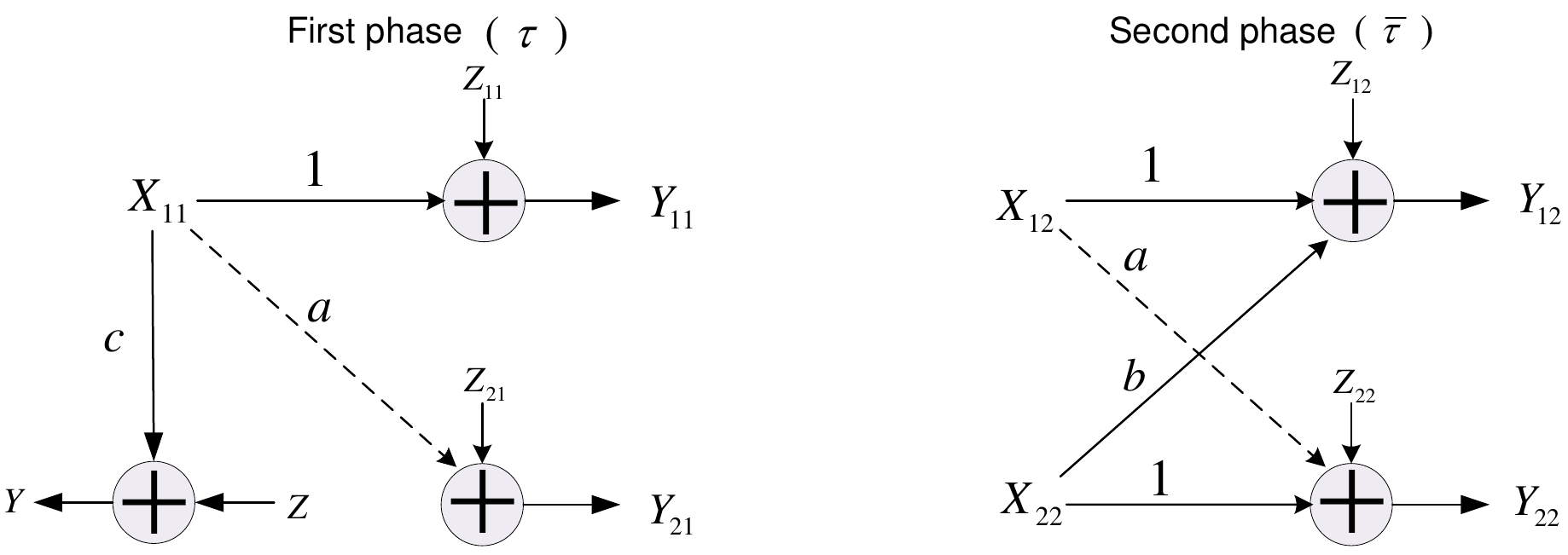}
\caption{The half-duplex Gaussian cognitive relay channel model.}
\label{fig:model_HD_CR}
\end{figure*}
\section{Half-duplex Gaussian CRC rate regions}
\subsection{ Half-duplex Gaussian CRC model}
The Gaussian model for the half-duplex cognitive relay channel is shown in Figure \ref{fig:model_HD_CR}. The input-output signals can be represented as
\begin{align}
\text{First phase :}\quad \;\, Y &= c X_{11} + Z, \nonumber\\
        Y_{11} &= X_{11} + Z_{11}, \nonumber\\
        Y_{21} &= a X_{11} + Z_{21};\label{eq:channel_first_phase}\\
\text{Second phase :} \quad Y_{12} &= X_{12} + b X_{22} + Z_{12}, \nonumber\\
        Y_{22} &= a X_{12} + X_{22} + Z_{22},\label{eq:channel_second_phase}
\end{align}
where $X_{11}$ is the transmit signal of $S_1$ in the first phase, $X_{12}$ and $X_{22}$ are the transmit signals of $S_1$ and $S_2$ in the second phase, respectively.
$Y$, $Y_{11}$ and $Y_{21}$ are the received signals at $S_2$, $T_1$ and $T_2$ in the first phase.
$Y_{21}$ and $Y_{22}$ are the received signals at $T_1$ and $T_2$ in the second phase. $a$, $b$, and $c$ are the channel gains where the direct links are normalized to $1$ as in the standard interference channel\cite{Carleial78}.
$Z$, $Z_{11}$, $Z_{21}$, $Z_{12}$, and $Z_{22}$ are independent white Gaussian noises with unit variance.

In the following section, we only provide the analysis for the half-duplex Gaussian HK-PDF-binning scheme and omit the analysis for half-duplex PDF-binning, which is a special case of HK-PDF-binning.

\subsection{Signaling and rates for the half-duplex HK-PDF-binning}\label{sec:HD_Gaussian_Combined}


In a Gaussian channel, input signals for the HK-PDF-binning scheme as in Section \ref{sec:HD_DMC_Combined} can be represented as
\begin{align}\label{eq:HD_Gaussian_Combined_signaling}
    X_{11} &= \alpha_1 S_{10}(w_{10}), \\
    X_{12} &= \alpha_2 S_{10}(w_{10}) + \beta_2  S_{11}(w_{11}) + \gamma_2 S_{12}(w_{12}), \nonumber \\
    X_{22} &= \theta S_{21}(w_{21}) + \mu\left(\rho S_{10}(w_{10}) +\sqrt{1-\rho^2}S_{22}\right),\nonumber \\
    U_{22} &= X_{22} + \lambda S_{10} =(\mu\rho+\lambda)S_{10}+\theta S_{21} +\mu\sqrt{1-\rho^2}S_{22},\nonumber
\end{align}
where $S_{10}$, $S_{11}$, $S_{12}$, $S_{21}$ and $S_{22}$ are independent $\mathcal{N}(0,1)$ random variables that encode $w_{10}$, $w_{11}$, $w_{12}$, $w_{21}$ and $w_{22}$ respectively,
$U_{22}$ is the Gelfand-Pinsker binning variable that encodes $w_{22}$.
The parameter $\rho$ is the correlation factor between the transmit signal $X_{22}$ and the state $X_{11}$, similar to that in Section \ref{sec:signal_FD_PDF_Binning}.
$\lambda$ is a parameter for binning.
Parameters $\alpha_1$, $\alpha_2$, $\beta_2$, $\gamma_2$, $\theta$ and $\mu$ are the corresponding power allocations that satisfy the power constraints
\begin{align}\label{eq:combined_power_cons}
    \tau \alpha_1^2 + \bar{\tau} (\alpha_2^2 + \beta_2^2 + \gamma_2^2) &\leq P_1, \nonumber\\
    \bar{\tau} (\mu ^2 + \theta^2) &\leq P_2,
\end{align}
where $\tau$ and $\bar{\tau}=1-\tau$ are the time duration for the two phases.

Substitute $X_{11}$, $X_{12}$, $X_{22}$ into $Y$, $Y_{11}$, $Y_{21}$, $Y_{12}$, $Y_{22}$ in \eqref{eq:channel_first_phase} and \eqref{eq:channel_second_phase}, we get
\begin{align}
    Y      &= c \alpha_1 S_{10} + Z, \nonumber   \\
    Y_{11} &= \alpha_1 S_{10} + Z_{11},\nonumber   \\
    Y_{21} &= a\alpha_1 S_{10} + Z_{21}, \nonumber   \\
    Y_{12} &=  (\alpha_2+b\mu\rho) S_{10} + \beta_2 S_{11} + \gamma_2 S_{12} + b\theta S_{21} 
            + b\mu\sqrt{1-\rho^2}S_{22} + Z_{12},\nonumber   \\
    Y_{22} &= (a\alpha_2+\mu\rho) S_{10} + a\beta_2  S_{11} + a\gamma_2 S_{12}+ \theta S_{21} 
              +  \mu\sqrt{1-\rho^2}S_{22} + Z_{22}.
\end{align}
\begin{cor}\label{cor:HD_Gaussian_Combined}
The achievable rate region for the half-duplex cognitive relay channel using Han-Kobayashi PDF-binning
 is  the convex hull of all rate pairs ($R_1$, $R_2$) satisfying
\begin{align}\label{eq:HD_Gaussian_Combined}
                             R_{1} &\leq \min\{I_2 + I_5, I_6\}, \nonumber\\
                             R_{2} &\leq I_{12}- I_1,\nonumber \\
                     R_{1} + R_{2} &\leq \min\{I_2 + I_7,I_8 \} + I_{13}- I_1,\nonumber \\
                     R_{1} + R_{2} &\leq \min\{I_2 + I_3,I_4\} + I_{14}- I_1,\nonumber \\
                     R_{1} + R_{2} &\leq \min\{I_2 + I_9, I_{10} \} + I_{11} - I_1, \nonumber\\
                    2R_{1} + R_{2} &\leq \min\{ I_2 + I_3,I_4\} + \min\{I_2 + I_9, I_{10} \} + I_{13}- I_1,\nonumber \\
                    R_{1} + 2R_{2} &\leq \min\{I_2 + I_7, I_8 \}+ I_{11} - I_1 + I_{14}- I_1,
\end{align}
where
\begin{align*}
    I_2 &= \tau C\left(c^2 \alpha_1^2 \right),\\
    I_3 &= \bar{\tau} C\left( \frac{\gamma_2^2}{b^2\mu^2(1-\rho^2) + 1} \right),\\
    I_4 &= \tau C\left(\alpha_1^2 \right) + \bar{\tau} C\left( \frac{(\alpha_2+b\mu\rho)^2 + \gamma_2^2 }{b^2\mu^2(1-\rho^2) + 1}\right), \\
    I_5 &= \bar{\tau}  C\left( \frac{\beta_2^2 + \gamma_2^2}{b^2\mu^2(1-\rho^2) + 1} \right),\\
    I_6 &= \tau C\left(\alpha_1^2 \right) + \bar{\tau} C\left( \frac{(\alpha_2+b\mu\rho)^2 + \beta_2^2 + \gamma_2^2 }{b^2\mu^2(1-\rho^2) + 1} \right),\\
    I_7 &= \bar{\tau} C\left( \frac{\gamma_2^2  + b^2\theta^2}{b^2\mu^2(1-\rho^2) + 1} \right),\\
    I_8 &= \tau C\left(\alpha_1^2 \right) + \bar{\tau} C\left( \frac{(\alpha_2+b\mu\rho)^2 + \gamma_2^2  + b^2\theta^2 }{b^2\mu^2(1-\rho^2) + 1}\right),\\
    I_9 &= \bar{\tau}  C\left( \frac{\beta_2 ^2 + \gamma_2 ^2  + b^2\theta^2}{b^2\mu^2(1-\rho^2) + 1} \right),\\
    I_{10} &= \tau C\left(\alpha_1^2 \right) 
            + \bar{\tau} C\left( \frac{(\alpha_2+b\mu\rho) ^2 + \beta_2 ^2 + \gamma_2 ^2  + b^2\theta ^2 }{b^2\mu^2(1-\rho^2) + 1} \right),\\
    I_{11}-I_{1} &=\bar{\tau} C \left( \frac {\mu^2(1-\rho^2) } {a^2\gamma_2^2 + 1} \right),\\
    I_{12} - I_1 &= \bar{\tau} C \left( \frac {\mu^2(1-\rho^2) } {a^2\gamma_2^2 + 1} \right) 
            + \bar{\tau} C\left( \frac{\theta^2}{(a\alpha_2+\mu\rho)^2 + a^2\gamma_2 ^2 +  \mu^2(1-\rho^2) + 1} \right),
\end{align*}
\begin{align*}
    I_{13} - I_1 &= \bar{\tau} C \left( \frac {\mu^2(1-\rho^2) } {a^2\gamma_2^2 + 1} \right) 
            + \bar{\tau} C\left( \frac{  a^2\beta_2 ^2}{(a\alpha_2+\mu\rho)^2 + a^2\gamma_2 ^2 +  \mu^2(1-\rho^2) + 1} \right),\\
    I_{14} - I_1 &= \bar{\tau} C \left( \frac {\mu^2(1-\rho^2) } {a^2\gamma_2^2 + 1} \right) 
             + \bar{\tau} C\left( \frac{ a^2\beta_2^2  + \theta^2 }{(a\alpha_2+\mu\rho)^2 + a^2\gamma_2 ^2 +  \mu^2(1-\rho^2) + 1} \right),
\end{align*}
and $C(x)=0.5\log_2(1+x)$; $\tau \in{[0,1]}$ and $\tau+\bar{\tau} = 1$;
$\rho \in{[-1,1]}$ is the correlation factor between $S_2$'s transmit signal $X_{22}$ and the state $X_{11}$;
and the power allocations $\alpha_1$, $\alpha_2$, $\beta_2$, $\gamma_2$, $\theta$ and $\mu$ satisfy the power constraints \eqref{eq:combined_power_cons}.
\end{cor}
\begin{proof}
    Applying the signaling in  \eqref{eq:HD_Gaussian_Combined_signaling} to Theorem \ref{thm:HD_DMC_combined},
    we obtain the rate region in Corollary \ref{cor:HD_Gaussian_Combined}.
\end{proof}
\begin{rem}{Inclusion of half-duplex PDF-binning and Han-Kobayashi schemes.}
\begin{itemize}
\item If we set $\tau=0$, $\alpha_1=\alpha_2=0$, $\rho=0$, rate region \eqref{eq:HD_Gaussian_Combined} becomes the Han-Kobayashi region \cite{Han}.
\item If we set $\beta_2=\theta=0$, rate region \eqref{eq:HD_Gaussian_Combined} becomes the half-duplex PDF-binning region.
\item
The half-duplex PDF-binning region is  the convex hull of all rate pairs ($R_1$, $R_2$) satisfying
\begin{align}
             R_{1} &\leq  \tau  C \left( { c^2 \alpha_1^2 }  \right) + \bar{\tau} C \left( \frac { \gamma_2^2 } {b^2 \mu^2(1-\rho^2)+ 1} \right) , \nonumber \\
             R_{1} &\leq  \tau  C \left( { \alpha_1^2 }  \right) + \bar{\tau} C \left( \frac {(\alpha_2+b\mu\rho)^2 +\gamma_2^2 } {b^2 \mu^2(1-\rho^2)+ 1} \right),\nonumber\\
             R_2 &\leq \bar{\tau} C \left( \frac {\mu^2(1-\rho^2) } {a^2 \gamma_2^2 + 1} \right),
\end{align}
where the power allocations $\alpha_1$, $\alpha_2$, $\gamma_2$ and $\mu$ satisfy the power constraints
\begin{align}\label{eq:bin_power_cons}
    \tau \alpha_1^2 + \bar{\tau} (\alpha_2^2 + \gamma_2^2) &\leq P_1, \nonumber\\
    \bar{\tau} \mu ^2 &\leq P_2.
\end{align}
\item
    The maximum rate for $S_1$ is achieved by setting $\beta_2=\theta=0$, $\rho = \pm 1$ and $\mu =\rho \sqrt{P_2}$ as
            \begin{align}\label{eq:HD_Gaussian_Bin_maxR1}
                R_{1}^{\max} = &\max_{\substack{\tau \alpha_1^2 + \bar{\tau} (\alpha_2^2 + \gamma_2^2) \leq P_1}}  \min 
                 \Bigg \{\tau C ({c^2 \alpha_1^2}) + \bar{\tau} C(\gamma_2^2), 
                                   \tau C(\alpha_1^2) + \bar{\tau} C\left( \left(\alpha_2+b\sqrt{P_2}\right)^2 + \gamma_2^2 \right) \Bigg \}.
            \end{align}
    A solution for this optimization problem is available in \cite{HostTIT2005}.
    Note that in the half-duplex mode, partial decode-forward achieves a strictly higher rate than pure decode-forward for the Gaussian channel.
\item
    The maximum rate for $S_2$ is achieved by setting $\tau=0$, $\rho=0$, $\alpha_1=\alpha_2=\beta_2=\gamma_2=\theta=0$, and $\mu=\sqrt{P_2}$ as
            \begin{align}\label{eq:HD_Gaussian_Bin_maxR2}
                R_{2}^{\max} = C(P_2).
            \end{align}
\end{itemize}
\end{rem}
\begin{rem}{The optimal binning parameter can be found similarly to the full-duplex case as follows.}
\begin{cor}
The optimal parameter $\lambda$ for the half-duplex Han-Kobayashi partial decode-forward binning scheme is
\begin{align}
    \lambda^* &= \frac{a\alpha_2\mu^2(1-\rho^2)-\mu\rho(a^2 \gamma_2^2+1)}{a^2 \gamma_2^2+\mu^2(1-\rho^2)+1}.
\end{align}
\end{cor}
\begin{proof}
Similar approach to the proof of Corollary \ref{cor:FD_Gaussian_optimalLambda}.
\end{proof}
\end{rem}

\subsection{Performance Comparison}\label{sec:numerical_sections}
\subsubsection{Existing results}
Very few results currently exist for the CRC. We can only find two results for the half-duplex mode. Next we briefly discuss each of these results.

Devroye, Mitran and Tarokh
\cite{Devroye2006TIT} propose four half-duplex protocols with rate region as the convex hull of the four regions.
One protocol is the Han-Kobayashi scheme for the interference channel, and the other three are 2-phase protocols in which $S_2$ obtains $S_1$'s message causally in the first phase as in a broadcast channel, then transmits cognitively in the second phase.
All these 3 protocols have $T_1$ decode at the end of both phases instead of only at the end of the second phase, hence they are suboptimal.
Protocol 2 has the idea of decode-forward by keeping the same input distribution at $S_1$ in both phases, but because in the second phase,
it reduces rate at $S_1$, thus it does not achieve the rate of decode-forward relaying.
Thus, even though the rate region includes the Han-Kobayashi region (in protocol 3), it does not include partial decode-forward relaying.

Chatterjee, Tong and Oyman \cite{ChatterjeeISIT2010} propose an achievable rate region for the half-duplex CRC based on rate-splitting, block Markov encoding, Gelfand-Pinsker binning and backward decoding.
The transmission is performed in $B$ blocks, each is divided into two phases.
In each phase, each user splits its message into two parts, one common and one private.
The primary user ($S_1$) superimposes its messages in both phases of the current block on the messages in the first phase of the previous block.
The cognitive user ($S_2$) only transmits in the second phase and bins both its message parts against the private message of $S_1$ in the first phase of the previous block.
Backward decoding is then used to decode the messages after $B$ blocks.
We have several comments on this scheme:
\begin{itemize}
\item
    Block Markovity is not necessary in half-duplex mode. We can superimpose the second-phase signal on the first-phase signal of the same block, instead of superimposing both phase signals on the first-phase signal of the previous block and using backward decoding as in \cite{ChatterjeeISIT2010}. Such a half-duplex block Markovity incurs unnecessarily long decoding delay and also wastes power to transmit the first-phase information of the current block, which is decoded backwardly.
\item
    Joint decoding of both the state and the binning auxiliary random variable at $T_1$ is not valid (similar to \cite{CaoISIT2007,YangTIT2011}). The rate region is thus larger than possible, but can be corrected in this step.
\item
    This scheme only covers half-duplex decode-forward relaying (when there is no binning) instead of partial decode-forward relaying and hence achieves a maximum rate for $R_1$ smaller than \eqref{eq:HD_Gaussian_Bin_maxR1}.
\end{itemize}

\subsubsection{Numerical Examples}
In this section, we provide numerical results to compare the two proposed schemes with the Han-Kobayashi and other known coding schemes \cite{Devroye2006TIT,ChatterjeeISIT2010} for the half-duplex CRC.

Figure \ref{fig:HD_PDFBin_HKPDFBin_comparison} shows the comparison between half-duplex PDF-binning, HK-PDF-binning and the Han-Kobayashi scheme.
It can be seen that although PDF-binning has a larger maximum rate for $R_1$ than the Han-Kobayashi scheme,
it is not always better.
But the half-duplex HK-PDF-binning rate region encompasses both the Han-Kobayashi and the PDF-binning regions.
\begin{figure}[t]
\centering
\includegraphics[scale=0.45]{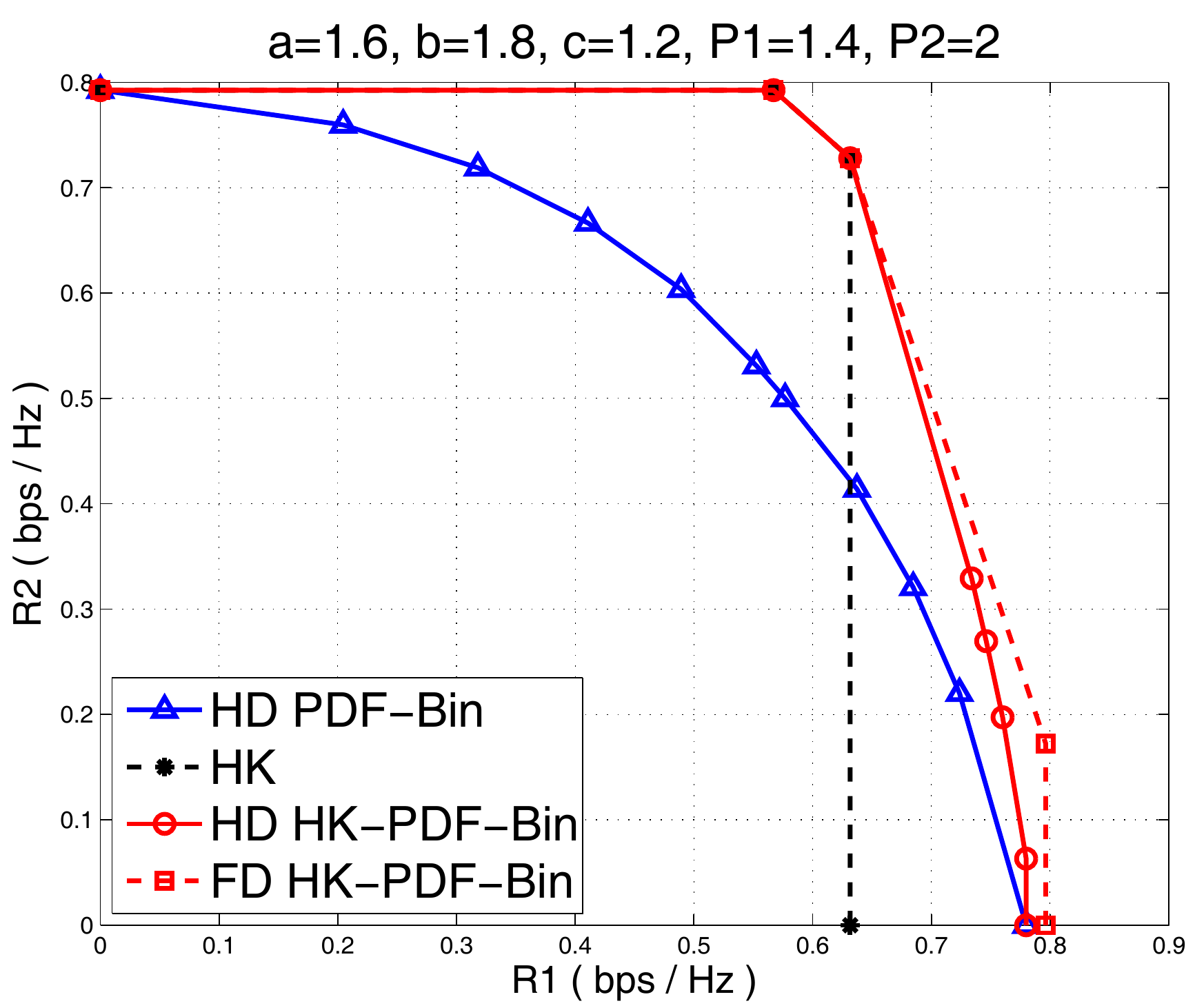}
\caption{Comparison of four coding schemes (HD = half-duplex, FD = full-duplex).}
\label{fig:HD_PDFBin_HKPDFBin_comparison}
\end{figure}

\begin{figure}[t]
\centering
\includegraphics[scale=0.45]{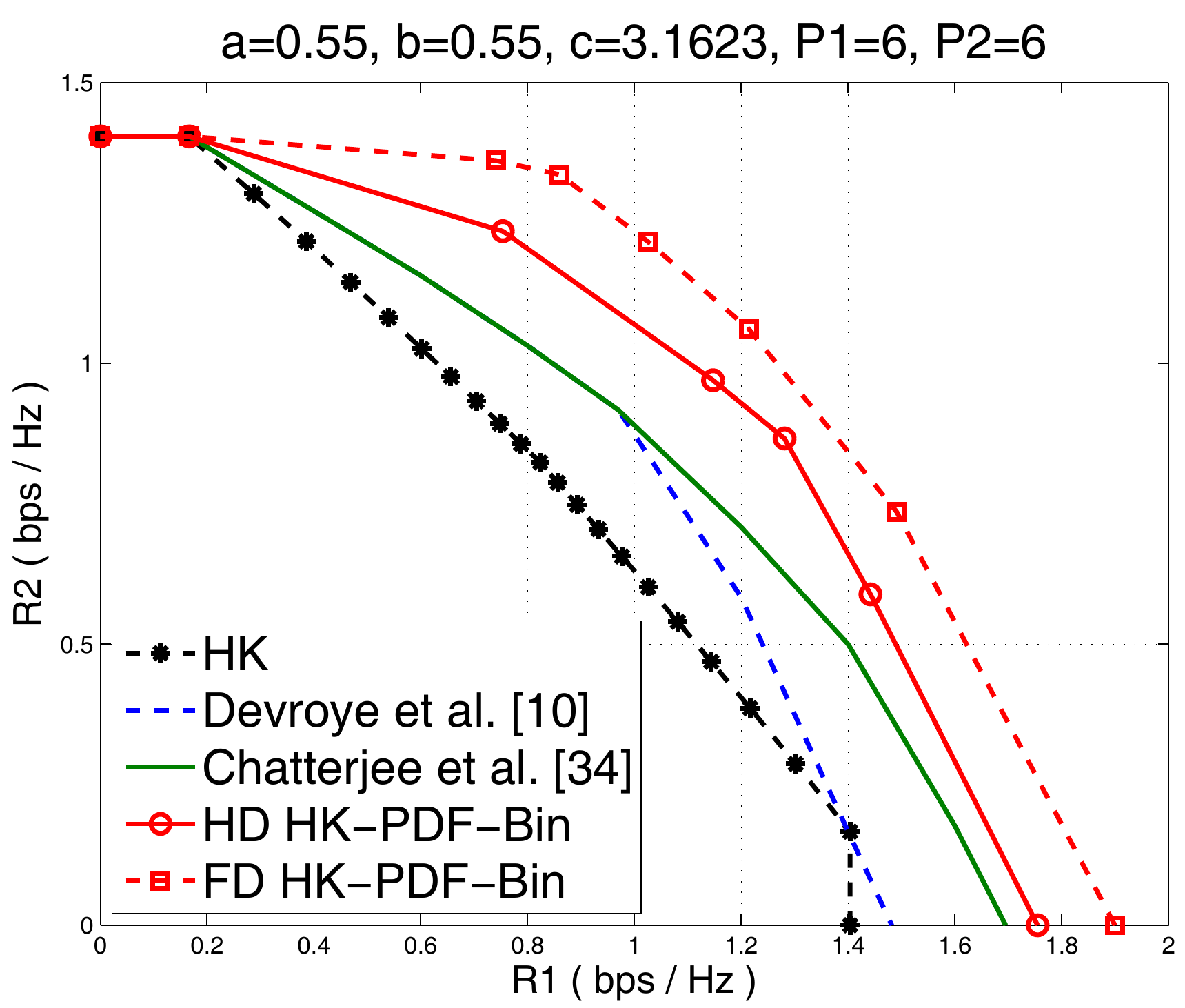}
\caption{Comparison of the HK-PDF-binning schemes with existing schemes.}
\label{fig:HD_HKPDFBin_Devroye_comparison}
\end{figure}

In Figure \ref{fig:HD_HKPDFBin_Devroye_comparison}, we compare the HK-PDF-binning schemes with existing half-duplex schemes for the CRC in \cite{Devroye2006TIT, ChatterjeeISIT2010}.
We can see that HK-PDF-binning is strictly better than
all existing schemes. Furthermore, the proposed scheme is more comprehensive than the protocols in \cite{Devroye2006TIT} and simpler than the scheme in \cite{ChatterjeeISIT2010}.

These figures also show that the gap in achievable rates by the HK-PDF-binning scheme in the half- and full-duplex modes is quite small.
Thus, the rate loss caused by the half-duplex constraint appears to be insignificant.

\section{Conclusion}
In this paper, we have proposed two new coding schemes for both the
full- and half-duplex cognitive relay channels. These two schemes are
based on partial decode-forward relaying, Gelfand-Pinsker binning and
Han-Kobayashi coding. The half-duplex schemes are adapted from the
full-duplex schemes by sending different message parts in different
phases, removing the block Markov encoding and applying joint decoding
across both phases.

When applied to Gaussian channels, different from the traditional
binning in dirty paper coding, in which the transmit signal is
independent of the state, here we introduce a correlation between the
transmit signal and the state, which enlarges the rate region by
allowing both binning and forwarding.  We also derive the optimal
binning parameter for each coding scheme.  Results show that the
proposed binning schemes achieve a higher rate than all existing
schemes for user $1$ by allowing user $2$ to also forward a part of
the message of user $1$.  Furthermore, the Han-Kobayashi PDF-binning
scheme contains both the Han-Kobayashi scheme and partial
decode-forward relaying and outperforms all existing schemes by
achieving a larger rate region for both users. Numerical results also
suggest that the difference in achievable rates between the half-
and full-duplex modes for the CRC is small.
\appendices
\section{Proof of the optimal binning parameter $\lambda^*$ for full-duplex HK-PDF-binning}\label{proof:FD_Gaussian_Combined_lambda}
To simultaneously maximize $R_1$ and $R_2$ in region \eqref{eq:FD_DMC_combined}, we can simply maximize the term $I_{11}-I_1$ as follows.
\begin{align*}
    &I(U_{22};Y_{2}|U_{21},U_{11}) - I(U_{22};T_{10}|U_{21}) \\
    &= H(Y_{2}|U_{21},U_{11}) - H(Y_{2}|U_{21},U_{22},U_{11}) 
     - H(U_{22}|U_{21}) + H(U_{22}|T_{10},U_{21})\\
    &= H(Y_{2}') - H(Y_{2}'|U_{22}') - H(U_{22}') + H(U_{22}|T_{10},U_{21})\\
    &= H(Y_{2}') +  H(U_{22}|T_{10},U_{21}) - H(U_{22}',Y_{2}'),
\end{align*}
where
\begin{align*}
    Y_{2}' &= Y_{2}|U_{21},U_{11} 
           = (a\alpha + \mu \rho) S_{10} ' + a\beta S_{10} + a\delta S_{12} 
            + \mu \sqrt{1-\rho^2}S_{22} + Z_2\\
    U_{22}' &= U_{22}|U_{21},U_{11} 
            = (\mu\rho+\lambda)S_{10}' + \mu\sqrt{1-\rho^2}S_{22}.
\end{align*}
Note that $\lambda$ only affects the last term $H(U_{22}',Y_{2}')$.
The covariance matrix between $U_{22}'$ and $Y_{2}' $ is
\begin{align}\label{eq:FD_Gaussian_Combined_cov}
\text{cov}(U_{22}', Y_{2}') =
    \begin{bmatrix}
    \text{var}(U_{22}')        & \text{E}(U_{22}',Y_{2}')\\
    \text{E}(U_{22}',Y_{2}')    & \text{var}(Y_{2}')
    \end{bmatrix},
\end{align}
where
\begin{align*}
    \text{var}(U_{22}') &= \mu  ^2 + \lambda ^2 + 2\mu\rho\lambda, \\
    \text{E}(U_{22}',Y_{2}')&=(\mu\rho+\lambda)(a\alpha+\mu\rho) + \mu^2(1-\rho^2),\\
    \text{var}(Y_{2}') &= (a\alpha+\mu\rho)^2+ a^2\beta^2 + a^2\delta^2+  \mu^2(1-\rho^2) + 1.
\end{align*}
Minimizing the determinant of the matrix in \eqref{eq:FD_Gaussian_Combined_cov} leads to
the optimal $\lambda$ as in \eqref{eq:lambda_FD_Gaussian_Combined}.

\section{Proof of Theorem \ref{thm:HD_DMC_Bin} (Half-duplex PDF-binning) }\label{proof:HD_DMC_Bin}
We use random codes and fix a joint probability distribution
\begin{align*}
    p(x_{11})p(x_{12}|x_{11})p(u_2|x_{11})p(x_{22}|x_{11},u_2).
\end{align*}

\subsection{Codebook generation}
\begin{itemize}
  \item
      Independently generate $2^{nR_{10}}$ sequences $x_{11}^n \sim \prod_{k=1}^{n}$ $p(x_{11k})$.
      Index these codewords as  $x_{11}^n(w_{10})$, $w_{10}\in [1,2^{nR_{10}}]$.

  \item
      For each $x_{11}^n(w_{10})$, independently generate $2^{nR_{11}}$ sequences
      $x_{12}^n \sim \prod_{k=1}^{n}p(x_{12k}|x_{11k}$).
       Index these codewords as $x_{12}^n(w_{11}|w_{10})$, $w_{11} \in [1,2^{nR_{11}}]$, $w_{10} \in [1,2^{nR_{10}}]$.

  \item
      Independently generate $2^{n(R_{2}+ R_{2}')}$ sequences $u_2^n \sim \prod_{k=1}^{n}p(u_{2k})$.
      Index these codewords as  $ u_2^n( w_{2}, v_{2})$, $w_{2} \in [1,2^{nR_{2}}]$
      and  $v_{2} \in [1,2^{nR_{2}'}]$.

  \item For each $x_{11}(w_{10})$
      and $u_2^n(w_{2}, v_{2})$, generate one $x_{22}^n \sim \prod_{k=1}^{n}p(x_{22i}|x_{11k}, u_{2k})$.
      Index these codewords as $x_{22}^n(w_{10}, w_{2}, v_2)$, $w_{2} \in [1,2^{nR_{2}}]$, $v_{2} \in [1,2^{nR_{2}'}]$.

\end{itemize}

\subsection{Encoding}
\begin{itemize}
  \item In the first phase, $S_1$ sends the codewords $x_{11}^{ \tau n}(w_{10})$. $S_2$ does not send anything.

  \item In the second phase, $S_1$ sends $x_{12}^{\bar{\tau}n}(w_{11}|w_{10})$.

      For $S_2$, it searches for a $v_{2}$ such that
            \begin{align}
                (x_{11}^{\bar{\tau}n}{(w_{10})}, u_2^{\bar{\tau}n}( w_{2}, v_{2}) ) \nonumber
                \in A_\epsilon^{(\bar{\tau}n)}(P_{X_{11} U_2}).
            \end{align}
      Such $v_{2}$ exists with high probability if
            \begin{align}
                R_{2}' \geq \bar{\tau} I( U_2; X_{11}).
            \end{align}
      $S_2$ then transmits $x_{22}^{\bar{\tau}n}(w_{10}, w_{2}, v_{2})$.
\end{itemize}
\subsection{Decoding}
\begin{itemize}
\item At the end of the first phase, $S_2$ searches for a unique $\hat{w}_{10}$ such that
\begin{align*}
    (x_{11}^{ \tau n}{(\hat{w}_{10})}, \mathbf{y})
    \in A_\epsilon^{(\tau n)}(P_{X_{11} Y}).
\end{align*}
We can show that the decoding error probability goes to 0 as $n\rightarrow \infty$ if
\begin{align}\label{eq:sp_bin_begin}
    R_{10}  &\leq  \tau  I(X_{11};Y).
\end{align}

\item At the end of the second phase, $T_1$ searches for a unique ($\hat{w}_{10},\hat{w}_{11}$) such that
\begin{align*}
    ( x_{11}^{\bar{\tau}n}{(\hat{w}_{10})}, x_{12}^{\bar{\tau}n}(\hat{w}_{11}|\hat{w}_{10}), \mathbf{y_{12}})&\in A_\epsilon^{(\bar{\tau}n)}(P_{X_{11} X_{12} Y_{12}})&\\
     \text{and}  \quad  (x_{11}^{ \tau n}{(\hat{w}_{10})},  \mathbf{y_{11}}) &\in A_\epsilon^{(\tau n)}(P_{X_{11} Y_{11}}).&
\end{align*}

Here $\mathbf{y_{11}}$ and $\mathbf{y_{11}}$ indicate the received vectors at $T_1$ during the first and second phases.
The decoding error probability goes to 0 as $n\rightarrow \infty$ if
\begin{align}
             R_{11} &\leq \bar{\tau} I(X_{12};Y_{12}|X_{11}) \nonumber \\
    R_{10} + R_{11} &\leq \bar{\tau} I(X_{11}, X_{12};Y_{12}) +  \tau  I(X_{11}; Y_{11}).
\end{align}
    \item $T_2$ treats the codeword $X_{11}^{\bar{\tau}n}$ from $S_1$ as the state and decodes $w_{2}$.
    It searches for a unique $\hat{w}_2$ for some $\hat{v}_{2} $  such that
\begin{align*}
    (u_2^{\bar{\tau}n}(\hat{w}_{2} ,\hat{v}_{2}), \mathbf{y_{22}}) \in   A_\epsilon^{(\bar{\tau}n)}(P_{U_2 Y_{22}}).
\end{align*}
The decoding error probability goes to 0 as $n\rightarrow \infty$ if
\begin{align} \label{eq:sp_bin_end}
             R_2 + R_2'  &\leq \bar{\tau} I(U_2;Y_{22}).
\end{align}
\end{itemize}
Combine all the above rate constraints, we get
\begin{align}
             R_{2}' &\geq \bar{\tau} I( U_2; X_{11})\nonumber \\
             R_{10}  &\leq  \tau  I(X_{11};Y)\nonumber \\
             R_{11} &\leq \bar{\tau} I(X_{12};Y_{12}|X_{11}) \nonumber \\
    R_{10} + R_{11} &\leq \bar{\tau} I(X_{11}, X_{12};Y_{12}) +  \tau  I(X_{11}; Y_{11}) \nonumber \\
             R_2 + R_2'  &\leq \bar{\tau} I(U_2;Y_{22}).
\end{align}
Let $R_1 = R_{10} + R_{11}$, apply Fourier-Motzkin Elimination, we get region \eqref{eq:DMC_PDF_Bin}.

\section{Proof of Theorem \ref{thm:HD_DMC_combined} (Half-duplex HK-PDF-binning) }\label{proof:HD_DMC_Combined}
We use random codes and fix a joint probability distribution
\begin{align*}
    &p(x_{11})p(u_{11})p(x_{12}|x_{11},u_{11})p(u_{21})
    p(u_{22}|u_{21},x_{11})p(x_{22}|x_{11},u_{21},u_{22}).
\end{align*}

\subsection{Codebook generation}
\begin{itemize}
  \item
      Independently generate $2^{nR_{10}}$ sequences $x_{11}^n \sim \prod_{k=1}^{n}p(x_{11k})$.
      Index these codewords as  $x_{11}^n(w_{10})$, $w_{10}\in [1,2^{nR_{10}}]$.

  \item
      Independently generate $2^{nR_{11}}$ sequences
      $u_{11}^n \sim \prod_{k=1}^{n}p(u_{11k}$).
       Index these codewords as $u_{11}^n(w_{11})$, $w_{11} \in [1,2^{nR_{11}}]$.

  \item
      For each $x_{11}^n(w_{10})$ and $u_{11}^n(w_{11})$, independently generate $2^{nR_{12}}$ sequences
      $x_{12}^n \sim \prod_{k=1}^{n}p(x_{12k}|x_{11k},u_{11k}$).
       Index these codewords as $x_{12}^n(w_{12}|w_{10},w_{11})$, $w_{12} \in [1,2^{nR_{12}}]$.

  \item
      Independently generate $2^{n(R_{21})}$ sequences $u_{21}^n \sim \prod_{k=1}^{n}p(u_{21k})$.
      Index these codewords as  $u_{21}^n( w_{21})$, $w_{21} \in [1,2^{nR_{21}}]$.

  \item
      For each $u_{21}^n( w_{21})$, independently generate $2^{n(R_{22}+R_{22}')}$ sequences
      $u_{22}^n \sim \prod_{k=1}^{n}p(u_{22k}|u_{21k}$).
       Index these codewords as $u_{22}^n(w_{22},v_{22}|w_{21})$, $w_{22} \in [1,2^{nR_{22}}]$, $v_{22} \in [1,2^{nR_{22}}]$.

  \item For each $x_{11}(w_{10})$, $u_{21}^n(w_{21})$
      and $u_{22}^n(w_{22}, v_{22}|w_{21})$, generate one $x_{22}^n \sim \prod_{k=1}^{n}p(x_{22k}|u_{22k},u_{21k},x_{11k})$.
      Index these codewords as $x_{22}^n(w_{10}, w_{21},w_{22}, v_{22})$.

\end{itemize}

\subsection{Encoding}
\begin{itemize}
  \item In the first phase, $S_1$ sends the codewords $X_{11}^{ \tau n}(w_{10})$. $S_2$ does not send anything.

  \item In the second phase, $S_1$ sends $x_{12}^{\bar{\tau}n}(w_{12}|w_{10},w_{11})$.

      For $S_2$, it searches for some $v_{22}$ such that
            \begin{align*}
                (x_{11}^{\bar{\tau}n}{(w_{10})},u_{21}^{\bar{\tau}n}( w_{21}), &u_{22}^{\bar{\tau}n}( w_{22}, v_{22}|w_{21})  )
                \in A_\epsilon^{(\bar{\tau}n)}(P_{X_{11} U_{22}|U_{21}}).
            \end{align*}
      Such $v_{22}$ exists with high probability if
            \begin{align}
                          R_{22}' &\geq \bar{\tau} I( U_{22}; X_{11}|U_{21}).
            \end{align}
      $S_2$ then transmits $x_{22}^{\bar{\tau}n}(w_{10},w_{21},w_{22}, v_{22})$.
\end{itemize}

\subsection{Decoding}
\begin{itemize}
\item At the end of the first phase, $S_2$ searches for a unique $\hat{w}_{10}$ such that
\begin{align*}
    (x_{11}^{ \tau n}{(\hat{w}_{10})}, \mathbf{y})
    \in A_\epsilon^{(\tau n)}(P_{X_{11} Y}).
\end{align*}
We can show that the decoding error probability goes to 0 as $n\rightarrow \infty$ if
\begin{align}\label{eq:hd_bin_begin}
             R_{10}  &\leq  \tau  I(X_{11};Y).
\end{align}

\item At the end of the second phase, $T_1$ searches for a unique ($\hat{w}_{10},\hat{w}_{11},\hat{w}_{12}$) for some $\hat{w}_{21}$ such that
\begin{align*}
    ( x_{11}^{\bar{\tau}n}{(\hat{w}_{10})},u_{11}^{\bar{\tau}n}{(\hat{w}_{11})},  x_{12}^{\bar{\tau}n}(\hat{w}_{12}|\hat{w}_{10},\hat{w}_{11}),
     u_{21}^{\bar{\tau}n}{(\hat{w}_{21})},\mathbf{y_{12}})&\in A_\epsilon^{(\bar{\tau}n)}(P_{X_{11} U_{11} X_{12} U_{21} Y_{12}})\\
     \text{and} \quad x_{11}^{ \tau n}{(\hat{w}_{10})}, \mathbf{y_{11}}) &\in A_\epsilon^{(\tau n)}(P_{ X_{11} Y_{11}}).
\end{align*}
The decoding error probability goes to 0 as $n\rightarrow \infty$ if
\begin{align}
                                        R_{12} &\leq \bar{\tau} I(X_{12};Y_{12}|X_{11},U_{11},U_{21}) \nonumber\\
                               R_{10} + R_{12} &\leq \tau I(X_{11};Y_{11}) 
                                                + \bar{\tau} I(X_{11},X_{12};Y_{12}|U_{11},U_{21})  \nonumber \\
                               R_{11} + R_{12} &\leq \bar{\tau} I(U_{11},X_{12};Y_{12}|X_{11},U_{21}) \nonumber \\
                      R_{10} + R_{11} + R_{12} &\leq \tau I(X_{11};Y_{11}) 
                                                + \bar{\tau} I(X_{11},U_{11},X_{12};Y_{12}|U_{21})  \nonumber \\
                               R_{12} + R_{21} &\leq \bar{\tau} I(X_{12},U_{21};Y_{12}|X_{11},U_{11}) \nonumber \\
                      R_{10} + R_{12} + R_{21} &\leq \tau I(X_{11};Y_{11}) 
                                                + \bar{\tau} I(X_{11},X_{12},U_{21};Y_{12}|U_{11})  \nonumber \\
                      R_{11} + R_{12} + R_{21} &\leq \bar{\tau} I(U_{11},X_{12},U_{21};Y_{12}|X_{11}) \nonumber \\
             R_{10} + R_{11} + R_{12} + R_{21} &\leq \tau I(X_{11};Y_{11}) 
                                                + \bar{\tau} I(X_{11},U_{11},X_{12},U_{21};Y_{12}).
\end{align}

    \item $T_2$ uses jointly decoding to decode $(w_{11},w_{21},w_{22})$.
    It searches for a unique $(\hat{w}_{21},\hat{w}_{22})$ for some $(\hat{w}_{11},\hat{v}_{22})$  such that
\begin{align*}
    (u_{11}^{\bar{\tau}n}(\hat{w}_{11} ),&u_{21}^{\bar{\tau}n}(\hat{w}_{21} ), u_{22}^{\bar{\tau}n}(\hat{w}_{22},\hat{v}_{22}|\hat{w}_{21} ),\mathbf{y_{22}}) 
    \in   A_\epsilon^{(\bar{\tau}n)}(P_{U_{11} U_{21} U_{22} Y_{22}}).
\end{align*}
The decoding error probability goes to 0 as $n\rightarrow \infty$ if
\begin{align}
                        R_{22} + R_{22}' &\leq \bar{\tau} I(U_{22};Y_{22}|U_{21},U_{11}) \nonumber\\
               R_{21} + R_{22} + R_{22}' &\leq \bar{\tau} I(U_{21},U_{22};Y_{22}|U_{11}) \nonumber\\
               R_{11} + R_{22} + R_{22}' &\leq \bar{\tau} I(U_{11},U_{22};Y_{22}|U_{21}) \nonumber\\
      R_{11} + R_{21} + R_{22} + R_{22}' &\leq \bar{\tau} I(U_{11},U_{21},U_{22};Y_{22}).
\end{align}
\end{itemize}
Let $R_1=R_{10}+R_{11}+R_{12}$ and $R_2=R_{21}+R_{22}$, apply Fourier-Motzkin Elimination  on the above constraints, we get region \eqref{eq:HD_DMC_combined}.

\bibliographystyle{IEEEtran}
\bibliography{references}




\end{document}